\documentclass[a4paper,onecolumn,unpublished]{quantumarticle}
\pdfoutput=1
\usepackage{graphicx}
\usepackage{amssymb}
\usepackage{amsmath}
\usepackage{amsthm}
\usepackage{amsfonts}
\usepackage{hyperref}
\usepackage{diagbox}
\usepackage{mathtools}
\usepackage{braket}
\usepackage{dsfont}
\usepackage{physics}
\usepackage{xfrac}
\usepackage{xcolor}
\usepackage{algorithm}
\usepackage{algpseudocode}
\usepackage{enumerate}
\usepackage{authblk}
\usepackage{tikz}
\usetikzlibrary{positioning,decorations.pathmorphing}
\usepackage{blkarray}
\usepackage[numbers,sort&compress]{natbib}
\usepackage{doi}

\newtheorem{lemma}{Lemma}
\newtheorem{theorem}{Theorem}
\newtheorem{definition}{Definition}
\newtheorem{corollary}{Corollary}
\newtheorem{remark}{Remark}

\newtheorem{proposition}{Proposition}

\newcommand{\im}{\text{im }}

\title{Fault-Resilience of Dissipative Processes for Quantum Computing}
\author{James Purcell}
\affiliation{Department of Computer Science, University College London, London WC1E 6EA, United Kingdom}
\email{j.purcell@cs.ucl.ac.uk}
\author{Abhishek Rajput}
\affiliation{Department of Computer Science, University College London, London WC1E 6EA, United Kingdom}
\affiliation{Phasecraft Ltd., London W1T 4PW, United Kingdom}
\author{Toby Cubitt}
\email{t.cubitt@ucl.ac.uk}
\affiliation{Department of Computer Science, University College London, London WC1E 6EA, United Kingdom}
\affiliation{Phasecraft Ltd., London W1T 4PW, United Kingdom}
\begin{document}
	\maketitle
	\begin{abstract}
		Dissipative processes have long been proposed as a means of performing computational tasks on quantum computers that may be intrinsically more robust to noise.
		In this work, we prove two main results concerning the error-resilience capabilities of two types of dissipative algorithms: dissipative ground state preparation in the form of the dissipative quantum eigensolver (DQE), and dissipative quantum computation (DQC).
		The first result is that under circuit-level depolarizing noise, a version of the DQE algorithm applied to the geometrically local, stabiliser-encoded Hamiltonians that arise naturally when fermionic Hamiltonians are represented in qubits, can suppress the additive error in the ground space overlap of the final output state exponentially in the code distance. This enables us to get closer to fault-tolerance for this task without the associated overhead.
		In contrast, for computation as opposed to ground state preparation, the second result proves that DQC is no more robust to noise than the standard quantum circuit model.

	\end{abstract}

	%\newpage
	%\tableofcontents
	%\newpage

	\section{Introduction}

        The biggest challenge to realising meaningful quantum computation is making those computations robust to the noise and errors inherent in real-world quantum hardware.
        Quantum error-correction and fault-tolerant circuit constructions provide the theoretical answer to this obstacle~\cite{gottesman1997stabilizer,mikeike}.
        The quantum threshold theorem proves that once the error rate is below a fixed threshold, arbitrarily large quantum computations can be performed to arbitrary accuracy \cite{knill-laflamme,aharonov1999faulttolerantquantumcomputationconstant}.
        However, whilst the overhead of fault-tolerant constructions in terms of physical qubits and gates theoretically scales only modestly (poly-logarithmically) with problem size, the constants are very large.
        It is well-known that the overhead for fully-scalable fault-tolerant quantum computation is likewise very large in practice, conservatively requiring tens of thousands of physical qubits for each logical qubit, and gate overheads even larger than this~\cite{Gidney_2021}.

        The majority of these results are derived in the unitary circuit model of quantum computation, which satisfies the DiVincenzo critera for universal quantum computation~\cite{DiVincenzo_2000}.
        However in 2009, \cite{VWC09} proved that purely dissipative dynamics can be used to perform universal quantum computation. 
        Their dissipative quantum computation (DQC) model side-steps some of the DiVincenzo requirements; in particular, it does not require reliable state initialisation.
        Indeed, as \cite{VWC09} pointed out, the output of a DQC computation is completely insensitive to the initial state. 
        This suggested that dissipative quantum algorithms and computation may be inherently more resilient to noise and errors than other methods of performing quantum computation, though this was not rigorously proven then.

       A major application for quantum computers is to probe the ground state properties of the many-body quantum systems that appear in fields such as condensed matter physics, materials science, and quantum chemistry. In fact, a key aim of quantum physics for many decades has been exactly or approximately solving for ground states of specific Hamiltonians modeling physical phenomena. Hamiltonians for which this program was successful include, among others, extensions of the 1D Heisenberg spin models \cite{bethe1931theorie,majumdar1969next,affleck2004rigorous}, extensions of the Ising model \cite{pfeuty1970one}, models of topological order \cite{kitaev2003fault,kitaev2006anyons}, and models of strongly correlated matter \cite{sachdev1993gapless}.
    	The physical insights explaining their success involve some of the most interesting developments in quantum information, such as entanglement area laws \cite{hastings2007area} and the notion of matrix product states \cite{Klumper_1993,GarciaMPS}.

	However, finding Hamiltonian ground states is a fundamentally hard problem in general.
	Even a relaxation---probing some specific property of the ground state---often also remains intractable.
	For example, approximating the ground state energy is NP-hard for classical Hamiltonians \cite{F_Barahona_1982} and QMA-hard for quantum Hamiltonians \cite{kitaev}.
	The latter hardness still holds even for Hamiltonians with very restricted structure \cite{Aharonov_2009}.
	Efficient preparation of a general Hamiltonian ground state therefore implies BQP$=$QMA and so is thought to take at least exponential time. (Preparing the corresponding quantum state given a sufficiently structured description of the ground state -- such as an MPS, PEPS, or stabiliser state -- is a much easier problem however \cite{VWC09}).

	Despite these complexity theoretic limitations, numerous algorithms for finding ground states have been developed and indeed often show success in some specific practical applications. In the same paper as they introduced DQC, \cite{VWC09} (and independently \cite{Kraus2008pyc}) also introduced the dissipative state engineering algorithm for preparing ground states of certain classes of quantum many-body Hamiltonians, namely frustration-free Hamiltonians of which stabiliser Hamiltonians and parent Hamiltonians of MPS and PEPS are special cases (note that the algorithm works for arbitrary frustration-free Hamiltonians but is only efficient for the latter special cases).
	The idea of dissipative state engineering is to design the system such that the interactions with the environment drive the system to a steady state.
	In the case of ground state preparation, the steady state should be designed to be close to the ground state of the Hamiltonian in question. \cite{VWC09} also hypothesised improved resilience to mid-algorithm noise as a result of the observation that under a single error, the initial state will still eventually converge to the fixed point of the map.
	This intuition though does not necessarily say anything about resilience to a positive \textit{rate} of noise.

    Recently however, \cite{cubitt2023dissipative} built on the dissipative state engineering results to show that the ground state of \emph{any} quantum many-body system can be prepared dissipatively using the dissipative quantum eigensolver algorithm that incorporates weak measurements.
    Moreover, they proved rigorously that the promise of inherent fault-resilience for dissipative algorithms holds here: the DQE algorithm, as well as the original dissipative state engineering algorithm, will succeed even if there is a positive rate of errors per time-step, as long as that error rate is below some finite threshold. This is weaker than full fault-tolerance though, as the error in the final state under a constant rate of noise cannot be made arbitrarily small for a general Hamiltonian (but in contrast to full fault-tolerant quantum computation, it incurs the minimal overhead of just a single additional ancilla qubit).

	These observations about dissipative quantum computation and algorithms raise the following questions: (1)~Is dissipative quantum computation truly more resilient to noise than the circuit model?
	(2)~Are there families of Hamiltonians with exploitable structures for which we can get closer to fault-tolerance without the associated overhead under the DQE algorithm?

	In Section \ref{section:dqestab} of this paper, we develop a specific version of the DQE algorithm for geometrically-local, \textit{stabiliser-encoded Hamiltonians} and show that under circuit-level depolarizing noise with noise rate $\delta$, the overlap of the algorithm's output state with the ground space is $1- \varepsilon - O(c^{(d+1)^2/4D_R})$ (see Theorem \ref{thm:noisysDQE}). Here $d$ is the distance of the stabiliser code, $c < 1$ and depends on the stabiliser code parameters and $\delta$, $D_R$ is the constant quantum circuit depth of the stabiliser recovery operation, and $\varepsilon < 1$ is an independent constant related to the strength of the weak measurements employed in the algorithm and can be made arbitrarily small. Compared to the standard DQE algorithm result of $1 - \varepsilon - O(\delta)$ (see Theorem 61 of \cite{cubitt2023dissipative}), this is therefore a further suppression in the additive error of the prepared state that is exponential in the code distance. This result also implies that if the noise rate $\delta$ is below a certain threshold (see the discussion after Theorem \ref{thm:noisysDQE}), the algorithm inherits additional fault-resilience from the stabiliser-code structure of the Hamiltonian. In the worst-case scenario however, its run time is exponential just like the conventional DQE algorithm (though it  often converges to the ground-state quickly and in low circuit-depth in practice; see   \cite{cubitt2023dissipative} for details of the practical implementation of the DQE algorithm). 

	In Section~\ref{section:dqc}, we turn our attention to dissipative quantum computation (DQC).
        Here, in contrast to dissipative state preparation, we prove that the tolerance of DQC to errors and noise is no better than the standard circuit model.
        More precisely, we show that a variant of the original Lindbladian formulation of DQC in~\cite{VWC09} that converges to the same fixed point in at most the same amount of time is exponentially close to a discrete-time formulation of the dynamics which corresponds, operationally, to a classical random walk on the corresponding quantum circuit.
        Under a natural noise model---namely, iid depolarising noise---the tolerance of DQC to noise can then readily be shown to be no better (indeed, in general worse) than the tolerance of the raw quantum circuit. 
        
        Appendices \ref{appendix:notation} and \ref{appendix:stabcodes} outline the concepts from stabiliser code theory and the notation used throughout the paper.

	\section{Dissipative Quantum Eigensolver for Stabiliser-Encoded Hamiltonians}

	\label{section:dqestab}

        Some of the most compelling practical applications of quantum simulation are to electronic structure problems in chemistry and materials science.
        As these Hamiltonians describe the states of electrons, they are fermionic Hamiltonians.
        As quantum computers are built out of qubits, simulating electronic structure Hamiltonians on a quantum computer requires mapping the fermionic Hamiltonian to a Hamiltonian on qubits.
        Many such fermionic-to-qubit mappings are known and the original and simplest is the Jordan-Wigner transformation, which directly maps fermionic Fock states to computational basis states of the qubits \cite{JordanWigner,Nielsen2005TheFC}.
        To preserve the anti-commutation relations of the fermionic operators under it however, most Hamiltonian terms involving constant numbers of fermionic modes have to be mapped to strings of Pauli operators that act across a large fraction of all the qubits.

        Since the advent of quantum information, a number of \emph{local} fermionic-to-qubit mappings have been developed to address this limitation so that Hamiltonian terms acting locally on the fermionic modes are mapped to Pauli terms acting on a constant number of qubits.
        Examples include the Verstraete-Cirac encoding~\cite{Verstraete_2005}, the Bravyi-Kitaev superfast encoding~\cite{Bravyi_2002}, the compact encoding~\cite{DerbyCompactfermion}, and others \cite{Jiang2019}.
        Any such local fermion mapping must necessarily map into a subspace of the full qubit Hilbert space and all of the above local fermion-to-qubit mappings have in common that the subspace mapped into forms a stabiliser subspace, usually with local stabiliser generators.
        Thus one of the most important classes of Hamiltonians in practical quantum simulation applications manifests naturally as qubit Hamiltonians with an additional structure that is closely related to quantum error-correcting codes.
        In a different context, Hamiltonians with this property also arise as boundary Hamiltonians of certain toy models of AdS/CFT~\cite{Apel_2022,Kohler_2019}.
        We will call such Hamiltonians, for which the ground state lies within the codespace of a stabiliser code, ``stabiliser-encoded Hamiltonians''.

        For general gapped many-body quantum Hamiltonians, there is no good reason to expect fault-resilience of dissipative ground state preparation beyond that proven in~\cite{cubitt2023dissipative}, which ultimately (albeit indirectly) stems from stability of the ground-state to perturbations that are small relative to the spectral gap of the Hamiltonian.
        Going beyond this for general ground-state preparation would likely require the dissipative algorithm itself to be made fault-tolerant.
        However, stabiliser-encoded Hamiltonians are a clear candidate for improved inherent error resilience in ground-state preparation algorithms without incurring the full overhead of fault-tolerant quantum computation.

	Another desideratum for fermion-to-qubit mappings is that local fermionic operators are mapped to local spin operators and this locality requirement ensures the resulting stabiliser-encoded Hamiltonians are geometrically local.
	We will therefore focus first on geometrically local stabiliser-encoded Hamiltonians and later discuss to what extent the results can be extended to other classes Hamiltonians.

	In Section \ref{subsec:AGSPconvergence}, we develop a version of the DQE algorithm adapted to preparing ground states of stabiliser-encoded Hamiltonians. We first show that it has the necessary ground space convergence properties in the ideal, noiseless scenario. Then in Section~\ref{subsec:faultresilience}, we prove that this algorithm inherits additional fault-resilience ``for free'' from the stabiliser structure, as compared to the original DQE algorithm, under circuit-level depolarizing noise. We primarily analyze the accuracy of the final state as opposed to the algorithm run time, which is exponential in the worst case scenario like the general DQE algorithm (though the DQE algorithm often succeeds in finding the ground state quickly in practice). 

	\subsection{Ground Space Convergence Properties}
	\label{subsec:AGSPconvergence}

	We consider a \textit{geometrically-local}, stabiliser-encoded Hamiltonian, which has the form:
	\begin{equation}
		H = H_0 + \beta \sum_s \Pi_s = \sum_{i=1}^m h_i + \beta \sum_{s=1}^k \Pi_s. \label{eq:hamiltonianFull}
	\end{equation}
	$H$ has a stabiliser part made up of $k$ stabiliser projectors $\Pi_s$ corresponding to $k$ independent stabiliser generators for an $[n,n-k,d]$ geometrically-local stabiliser code (see \ref{appendix:stabcodes} for the relevant notation), and non-stabiliser terms $h_i$ such that $$[h_i,\Pi_s] = 0 \ \forall i,s \ \text{ and } \  [\Pi_s,\Pi_{s'}] = 0 \ \ \forall s,s'.$$  The commuting constraints imposed upon the summands of $H$ imply that its ground space lies in the stabiliser codespace.
     
     We can also rescale the stabiliser projector terms by some number $\beta \geq 0$ if desired. Lemma \ref{lemma:codespaceproj} implies this part of the preceding Hamiltonian imposes an energy penalty on those states that do not lie within the codespace, so this rescaling effectively controls the strength of the perturbation to $H_0$. The choice of rescaling however \textit{does not} affect the validity of the subsequent results and merely affects the form that certain parameters may take. Regardless of the rescaling made, there is the caveat the ground space of the $H$ is an eigenspace corresponding to \textit{some} eigenvalue of $H_0$, not necessarily the lowest eigenvalue of $H_0$. This is all that is needed in the proof of Theorem \ref{thm:AGSP} to follow.

	Now consider the following operator:
	\begin{equation}
		K = \left(\prod_{s=1}^k (I - \Pi_s) \right) \left(\prod_{i=1}^m ((1-\varepsilon) I + \varepsilon k_i)\right) \left(\prod_{i=m}^1 ((1-\varepsilon) I + \varepsilon k_i)\right) \label{eq:AGSP}
	\end{equation}

	where $$k_i = \frac{(\|h_i\|/\kappa)(I - h_i/\|h_i\|)}{2}, \hspace{1em} \kappa = \sum_{i=1}^m \|h_i\| + \beta\sum_{s=1}^k \|\Pi_s\| = \sum_{i=1}^m \|h_i\| + \beta k.$$ 
    
    This operator essentially consists of projective measurements of the codespace projector along with weak-measurements of the terms in $H_0$ (up to rescalings and shifts) in \ref{eq:hamiltonianFull}, where the ``weakness" is determined by the parameter $\varepsilon$. Setting $\varepsilon = 1$ recovers the case where all terms in the stabiliser-encoded Hamiltonian are projectively measured while setting $\varepsilon = 0$ turns off the measurements of the terms in $H_0$.
     
    $K$ is manifestly Hermitian and we would like to show that it is an approximate ground state projector (AGSP) as defined below. This AGSP will be an important ingredient in the ground state preparation algorithm we outline later.

	\begin{definition}[\textbf{Definition 6 of \cite{cubitt2023dissipative}}]
		\label{def:AGSP}
		Let $\Pi_0$ be the projector onto the ground space of a Hamiltonian. A Hermitian operator $K$ is a $(\Delta, \Gamma, \delta)$-AGSP for $\Pi_0$ if there exists a projector $\Pi$ such that:
		\begin{enumerate}[(i)]
			\item $\|\Pi - \Pi_0\| \leq \delta$
			\item $[K,\Pi] = 0$
			\item $K\Pi \geq \sqrt{\Gamma} \Pi$
			\item $\|(I-\Pi)K(1-\Pi)\| \leq \sqrt{\Delta}$
		\end{enumerate}
	\end{definition}

	We will take $\Pi = \Pi_0$ unless otherwise specified, so the first property is satisfied with $\delta = 0$ in this situation. The third and fourth properties simply quantify how much the AGSP pushes a state into the ground space than out of it respectively. A ``good" AGSP, i.e. one that pushes a state into the ground space more than out of it, is therefore characterized by having $\sqrt{\Gamma} > \sqrt{\Delta}$. The last property can equivalently be written as $K(I - \Pi) \leq \sqrt{\Delta}(I - \Pi)$ if desired. 

    A concrete example of an AGSP is as follows: 

    \begin{lemma}[\textbf{Lemma 12 of \cite{cubitt2023dissipative}}]
        \label{lemma:AGSPexample}
        If $H = \sum_i h_i$ is a Hamiltonian, then $$K = \frac{I - H/\kappa}{2} = \sum_i \kappa_i k_i$$ where $$k_i = \frac{I - h_i/\|h_i\|}{2}, \quad \kappa = \sum_i \|h_i\|, \quad \kappa_i = \frac{\|h_i\|}{\kappa}$$ is a $(\Delta, \Gamma, 0)$-AGSP for the projector $\Pi_0$ onto the ground space of $H$ with $$\sqrt{\Gamma} = \frac{1 - \lambda_0/\sum_i \|h_i\|}{2}, \quad \sqrt{\Delta} = \frac{1 - \lambda_1/\sum_i \|h_i\|}{2},$$ where $\lambda_0, \lambda_1$ are the minimum and next-lowest eigenvalues of $H$ (not counting degeneracies). 
    \end{lemma}

    We will need to consider perturbations to AGSPs in what follows and therefore note the following relevant lemma.

    \begin{lemma}[\textbf{Lemma 11 of \cite{cubitt2023dissipative}}]
        \label{lemma:AGSPpert}
        Let $K$ be a $(\Delta, \Gamma, 0)$-AGSP for $\Pi_0$ with ground state degeneracy $N \coloneqq \tr \Pi_0$. If $K'$ is a Hermitian operator such that $\delta \coloneqq \|K - K'\| < |\sqrt{\Gamma}-\sqrt{\Delta}|$, then $K'$ is a $(\Delta + \delta, \Gamma - \delta, \varepsilon)$-AGSP with $$\varepsilon = \frac{2\sqrt{N}\delta}{\sqrt{\Gamma}-\sqrt{\Delta}}$$
    \end{lemma}
    We now prove the following: 
	\begin{theorem}
		Let $H$ be the stabilizer-encoded Hamiltonian of \eqref{eq:hamiltonianFull}, $K$ the Hermitian operator in \eqref{eq:AGSP}, and $\Pi$ the ground space projector of $H$. Suppose $\lambda$ is an eigenvalue of $H_0$ such that $H \Pi = \lambda \Pi$ and let $\lambda'$ be the next-largest eigenvalue of $H$. Then $K$ is an $(\Delta + O(\varepsilon^2),\Gamma - O(\varepsilon^2),O(\varepsilon^2))$-AGSP for $H$ where $$\sqrt{\Gamma} \coloneqq (1-\varepsilon)^{2m-1}\left[(1-\varepsilon) + \varepsilon \frac{\kappa_0}{\kappa} \left(1 -\frac{\lambda}{\kappa_0} \right) \right],$$ $$\sqrt{\Delta} \coloneqq (1-\varepsilon)^{2m-1} \left[(1-\varepsilon) + \varepsilon \frac{\kappa_0}{\kappa} \left(1-\frac{\lambda'}{\kappa_0}\right)\right],$$ and $$\kappa = \sum_{i=1}^m \|h_i\| + \beta\sum_{s=1}^k \|\Pi_s\| = \sum_{i=1}^m \|h_i\| + \beta k = \kappa_0 + \beta k.$$
        \label{thm:AGSP}
	\end{theorem}
	\begin{proof}
		Recall that $$K = \left(\prod_s (I - \Pi_s) \right) \left(\prod_{i=1}^m ((1-\varepsilon) I + \varepsilon k_i)\right) \left(\prod_{i=m}^1 ((1-\varepsilon) I + \varepsilon k_i)\right)$$ and define $$\tilde{K} = \left(\prod_s (I - \Pi_s) \right) \left [(1-\varepsilon)^{2m}I + 2\varepsilon(1-\varepsilon)^{2m-1}\sum_{i=1}^m k_i \right].$$ First observe that $$\sum_{i=1}^m k_i = \frac{1}{2\kappa} \left(\sum_{i=1}^m \|h_i\|I - h_i\right) = \frac{1}{2\kappa}(\kappa_0 I - H_0) = \frac{\kappa_0}{\kappa}K_0,$$ where by Lemma \ref{lemma:AGSPexample} $$K_0 = \frac{I - H_0/\kappa_0}{2}, \hspace{1em} \kappa_0 = \sum_{i=1}^m \|h_i\|$$ is a $(\Delta_0, \Gamma_0, 0)$-AGSP for the ground space projector $\Pi_0$ of $H_0$ with $$\sqrt{\Delta_0} = \frac{1 - \lambda_1/\kappa_0}{2}, \hspace{1em} \sqrt{\Gamma_0} = \frac{1- \lambda_0/\kappa_0}{2}$$ and $\lambda_0,\lambda_1$ the minimum and next-lowest eigenvalues of $H_0$ respectively.

		If $\Pi$ is the projector onto the ground space of $H$, it immediately follows that $[\Pi_s, \Pi] = 0$ for all $s$ and $[\Pi,k_i] = 0$ for all $i$. Thus $[\tilde{K},\Pi] = 0$.

		We now verify that $\tilde{K}\Pi \geq \sqrt{\Gamma} \Pi$ for some $\Gamma$ or equivalently that $\|\tilde{K}\Pi\| \geq \sqrt{\Gamma}$. Since $\Pi H_0 = \lambda \Pi$, where $\lambda$ is an eigenvalue of $H_0$ such that $\lambda \geq \lambda_0$, $\Pi K_0 =(1/2 - \lambda/2\kappa_0)\Pi$ and
		\begin{align*}
			\tilde{K}\Pi = \prod_s (I - \Pi_s) \left[(1-\varepsilon)^{2m} + 2\varepsilon(1-\varepsilon)^{2m-1} \frac{\kappa_0}{2\kappa} \left(1 -\frac{\lambda}{\kappa_0} \right) \right]\Pi.
		\end{align*}
		Since the ground space of $H$ lies inside the codespace, $$\left(\prod_{s}(I - \Pi_s)\right) \Pi = \Pi.$$ Then
		\begin{align}
			\|\tilde{K}\Pi\| = (1-\varepsilon)^{2m-1}\left[(1-\varepsilon) + \varepsilon \frac{\kappa_0}{\kappa} \left(1 -\frac{\lambda}{\kappa_0} \right) \right] \coloneqq \sqrt{\Gamma}. \label{eq:gammaAGSP}
		\end{align}
		Next, we have that $$(I - \Pi)K_0 = \frac{(I-\Pi)-(I - \Pi)H_0/\kappa_0}{2} \leq \left(\frac{1 - \lambda'/\kappa_0}{2}\right)(I - \Pi),$$ where $\lambda'$ is the next largest eigenvalue of $H_0$ after $\lambda$. We then have
		\begin{align*}
			\|(1-\Pi)\tilde{K}(1-\Pi)\| &\leq  \norm{\prod_s (I - \Pi_s)} \times \\
			&\left((1-\varepsilon)^{2m}\norm{(1-\Pi)(1-\Pi)} +2\varepsilon(1-\varepsilon)^{2m-1}\frac{\kappa_0}{\kappa} \norm{(1-\Pi)K_0(1-\Pi)}\right) \notag \\
			&\leq (1-\varepsilon)^{2m-1} \left[(1-\varepsilon) + \varepsilon \frac{\kappa_0}{\kappa} \left(1-\frac{\lambda'}{\kappa_0}\right)\right] \coloneqq \sqrt{\Delta} \label{eq:deltaAGSP}
		\end{align*}
		It follows that $\tilde{K}$ is a $(\Delta,\Gamma,0)$-AGSP with $\Delta, \Gamma$ as above.

		Now note that $$K = \left(\prod_s (I - \Pi_s) \right) \left((1-\varepsilon)^{2m}I + 2\varepsilon(1-\varepsilon)^{2m-1}\frac{\kappa_0}{\kappa}K_0 + O(\varepsilon^2)\right),$$ so $\|K - \tilde{K}\| = O(\varepsilon^2)$. It then follows from Lemma \ref{lemma:AGSPpert} above that $K$ is an $(\Delta+O(\varepsilon^2),\Gamma-O(\varepsilon^2),O(\varepsilon^2))$-AGSP.
	\end{proof}

	The preceding theorem also establishes that our choice of $K$ is in particular a good-AGSP.

	AGSPs can be thought of as Kraus operators of a completely positive (CP) map on bounded operators such that repeated iterations of this map increase the accuracy of the projection onto the ground space of the Hamiltonian. However, such a map is not necessarily trace preserving and only describes a particular set of measurement outcomes. We must therefore extend this map to a quantum instrument, i.e. a collection of CP maps that sum to a CP and trace-preserving (TP) map and that describes all possible measurement outcomes involved.

	In our scenario, we specifically wish to
	\begin{enumerate}
		\item Measure all syndrome space projectors $$P_x = \frac{\prod_{i=1}^k (I + (-1)^{x_i} g_i)}{2^k}, \hspace{1em} x = (x_1, \ldots, x_k) \in \mathbb{Z}_2^k$$ in \eqref{eq:syndromeprojs} on an input state $\rho \in \mathcal{B}(\mathcal{H})$.
		\item If measurements of $P_x$ for $x \neq 0$ give non-zero outcomes, apply rotations $U_x$ to rotate the components of $\rho$ not in the codespace back into it. There always exists a (not necessarily unique) unitary $U_x$ in the Pauli group such that \begin{equation}
			U_x P_x U_x^{\dag} = P_0, \label{eq:Ux-relation}
		\end{equation}so the $U_x$ are chosen in this manner.
		\item Weak measure all non-stabiliser terms $h_i$ in the Hamiltonian (i.e. apply the AGSP in \eqref{eq:AGSP}) and resample the initial state to a maximally mixed state on the codespace if any fail.
	\end{enumerate}
	The quantum instruments that accomplish the above steps are as follows:
	\begin{align}
		\mathcal{E}_{0}(\rho) &= K P_0 \rho P_0 K + \sum_{x \neq 0} (K U_x P_x \rho P_x U_x^{\dag} K) \label{eq:E0instrument}\\
		\mathcal{E}_{1}(\rho) &= \left(1 - \sum_x \tr(K U_x P_x \rho P_x U_x^{\dag} K)\right) \frac{P_0}{2^{n-k}} \label{eq:E1instrument}
	\end{align}

	where the sum in the second term of $\mathcal{E}_0$ is over all syndrome outcomes $x$ that are not the all 0 bit string and where we've defined $U_0 = I$ for the sum over all $x$ in $\mathcal{E}_1$. $\mathcal{E}_0$ represents the possible outcomes when the stabiliser measurements either fail or succeed but the weak measurements succeed. $\mathcal{E}_1$ represents those outcomes where the stabiliser measurements fail or succeed but the weak measurements fail, upon which we resample to the maximally mixed state on the codespace. It is clear that $\mathcal{E}(\rho) = \mathcal{E}_0(\rho) + \mathcal{E}_1(\rho)$ as defined is trace preserving, i.e. $\tr(\mathcal{E}(\rho)) = \tr \rho$.

	\begin{remark}
		{\normalfont We can write $\mathcal{E}_0 = \mathcal{K}\circ \mathcal{R}$ where $$\mathcal{R}(\rho) = P_0 \rho P_0 + \sum_{x \neq 0} U_x P_x \rho P_x U_x^{\dag}$$ and $$\mathcal{K}(\rho) = K P_0 \rho P_0 K.$$ The latter definition for the completely positive map $\mathcal{K}$, which is based on \eqref{eq:AGSP}, is technically not correct as this does not represent all the measurement outcomes. The actual definition should be $$\mathcal{K}(\rho) = K P_0 \rho P_0 K + \sum_{x \neq 0} P_x \rho P_x.$$ However, note that $$P_0 \mathcal{R}(\rho) P_0 = P_0 P_0 \rho P_0 P_0 + \sum_{x \neq 0} P_0 U_x P_x \rho P_x U_x^{\dag} P_0 = P_0 \rho P_0 + \sum_{x \neq 0} U_x P_x \rho P_x U_x^{\dag} = \mathcal{R}(\rho),$$ where the second to last equality follows from the relation that $U_x P_x = P_0 U_x$ in \ref{eq:Ux-relation}. Thus $$\mathcal{K}(\mathcal{R}(\rho)) = K P_0 \mathcal{R}(\rho) P_0 K + \sum_{x \neq 0} P_x \mathcal{R}(\rho) P_x = K \mathcal{R}(\rho)K = K P_0 \mathcal{R}(\rho) P_0 K.$$ Therefore we can effectively treat $\mathcal{K}$ as defined by $\mathcal{K}(\rho) = KP_0 \rho P_0 K$ provided it follows after the application of the recovery map $\mathcal{R}$. This will always be the case in what follows unless otherwise specified. $\blacksquare$}
	\end{remark}

	\begin{proposition}
		Let $\rho \in \mathcal{B}(\mathcal{H})$ be a density matrix on a Hilbert space $\mathcal{H}$ and $\mathcal{E}(\rho) = \mathcal{E}_0(\rho) + \mathcal{E}_1(\rho)$ the CPTP map with instruments $\mathcal{E}_0$ and $\mathcal{E}_1$ as defined in equations \eqref{eq:E0instrument}-\eqref{eq:E1instrument}. Then $\mathrm{tr}(P_0 \mathcal{E}^t(\rho)P_0) = 1$ for all $t  \geq 1$.
	\end{proposition}

	\begin{proof}
		Note that since $\tr(P_0) = 2^{n-k}$, $$\tr(P_0 \mathcal{E}_1(\rho)P_0) = 1 - \sum_x \tr(K U_x P_x \rho P_x U_x^{\dag} K).$$ From the relations $[K, P_0] = 0$, $U_x P_x U_x^{\dag} = P_0$, and the idempotency $P_x^2 = P_x$ of the projectors $P_x$, it follows that
		\begin{align*}
			\tr(P_{0} \mathcal{E}_0(\rho)P_0) &= \tr(P_0 K P_0 \rho P_0 K P_0) + \sum_{x \neq 0} \tr(P_{0} K U_x P_x \rho P_x U_x^{\dag} K P_0) \\
			&= \tr(K P_0 P_0 \rho P_0 P_0 K) + \sum_{x \neq 0} \tr(K U_x P_x P_x \rho P_x P_x U_x^{\dag} K) \\
			&= \tr(K P_0 \rho P_0 K) + \sum_{x \neq 0} \tr(K U_x P_x \rho P_x U_x^{\dag} K) \\
			&= \sum_x \tr(K U_x P_x \rho P_x U_x^{\dag} K).
		\end{align*}

		Thus $\tr(P_0(\mathcal{E}_0 + \mathcal{E}_1)P_0) = 1$. Suppose by induction that $\tr(P_0 \mathcal{E}^{t-1}(\rho)P_0) = 1$ for any $\rho$. Then $\tr(P_0 \mathcal{E}^{t-1}(\mathcal{E}(\rho))P_0) = 1$, so this implies that $ \tr(P_{0} \mathcal{E}^t(\rho)P_0) = 1$ for all $t$.
	\end{proof}

	\begin{proposition}
    	\label{prop:GSoverlap}
		Let $K$ be the AGSP defined in equation \eqref{eq:AGSP}, $D$ the dimension of the Hilbert space $\mathcal{H}$, $\Pi$ the ground space projector of $H$ in \eqref{eq:hamiltonianFull}, and $\mathrm{tr}(\Pi) = N$ the ground space degeneracy. Consider the stopped process where we iterate the instruments $\{\mathcal{E}_0, \mathcal{E}_1\}$ in \eqref{eq:E0instrument}-\eqref{eq:E1instrument} above starting from the maximally mixed state on the codespace $\rho_0 = P_0/2^{n-k}$ until we obtain a sequence of $n$ zeroes. The state $\rho_n$ at the stopping time then satisfies $$\tr(\Pi \rho_n) \geq 1 - \frac{1-\tr(\Pi \cdot \rho_0)}{\tr (\Pi \cdot \rho_0)} \left(\frac{\Delta}{\Gamma} \right)^n \geq 1 - \frac{D}{N}\left(\frac{\Delta}{\Gamma} \right)^n.$$
	\end{proposition}
	\begin{proof}
		Recall that from the definition of an AGSP, $[K,P_0] = 0$. This implies that for any state $\rho$,
		\begin{align}
			\mathcal{E}_0(\mathcal{E}_0(\rho)) &= \sum_{x'} \sum_x K P_0 U_{x'} P_{x'} (K P_0 U_x P_x \rho P_x U_x^{\dag} P_0 K) P_{x'} U_{x'}^{\dag} P_0 K \\
			&= K P_0 \left(\sum_x K U_x P_x \rho P_x U_x^{\dag}K\right) P_0 K = K P_0 \mathcal{E}_0(\rho) P_0 K
		\end{align}
		where in the second equality, we have used that $[K,P_0]=0$ and $P_0 P_x = 0$ implies that the outer sum is zero except for $x' = 0$. For our choice of $\rho_0$ in particular, the same relations imply $\mathcal{E}_0(\rho_0) = K \rho_0 K$. Thus from the definition of an AGSP in \ref{def:AGSP}, we have $$\tr(\Pi \mathcal{E}_0^n(\rho_0)) = \tr(\Pi (KP_0)^n \rho_0 (K P_0)^n) = \tr((K \Pi)^n \rho_0 (\Pi K)^n) \geq \Gamma^n \tr(\Pi \cdot \rho_0).$$ Similarly, $$\tr((1-\Pi)\mathcal{E}_0^n(\rho_0)) = \tr((I-\Pi)K^n \rho_0 K^n) \leq \Delta^n (1 - \tr(\Pi \cdot \rho_0)).$$ The proof then proceeds in the exact same way as Theorem 18 of \cite{cubitt2023dissipative}.
	\end{proof}

	\subsection{Fault Resilience}
	\label{subsec:faultresilience}

	We now prove that we obtain additional fault resilience from the DQE algorithm for stabiliser encoded Hamiltonians under circuit level depolarizing noise. Unless otherwise specified, we will work in a basis of our Hilbert space given by the syndrome outcomes, with the all 0 syndrome corresponding to the codespace projector $P_0$ coming first.

	The transfer matrices $E_0$ and $E_1$ corresponding to the channels $\mathcal{E}_0$ and $\mathcal{E}_1$ (see Appendix \ref{appendix:notation}) are
	\begin{align}
		E_0 &= (K \otimes \bar{K}) \left(P_0 \otimes \bar{P}_0 + \sum_{x \neq 0} U_x P_x \otimes \bar{U}_x \bar{P}_x \right) \\
		E_1 &= \frac{1-\sum_x \tr(K U_x P_x \rho P_x U_x^{\dag} K)}{2^{k}} \sum_{i,j} \ket{ii}\bra{jj},
	\end{align}
	where the sum over $i,j$ in $E_1$ runs up to the dimension of the codespace.

	Due to equation \eqref{eq:Ux-relation}, we can equivalently write the transfer matrix $E_0$ as $$E_0 = (K \otimes \bar{K}) \left(P_0 \otimes \bar{P}_0 + \sum_{x \neq 0} P_0 U_x \otimes \bar{P}_0 \bar{U}_x \right) = \text{DQE} \cdot R$$ where
	\begin{equation}
		\text{DQE} = K \otimes \bar{K}  \label{eq:DQEtransfer}
	\end{equation}
	and
	\begin{equation}
		R = P_0 \otimes \bar{P}_0 + \sum_{x \neq 0} P_0 U_x \otimes \bar{P}_0 \bar{U}_x. \label{eq:stabrecov}
	\end{equation}
	For our noise model, we will work with iid depolarizing noise with error probability $\delta$ which has the following form for $n$ qubits:
	\begin{equation}
		\mathcal{N}^{\otimes n}_{\delta}(\rho) = (1-\delta)^n \rho + \sum_{i=1}^n (1-\delta)^{n-i} \left(\frac{\delta}{3}\right)^i \left(\sum_{s: \text{weight s} = i} P(s) \rho P(s) \right).  \label{eq:iidDepolar}
	\end{equation}

	Here $P(s)$ denotes a Pauli string of weight $s$ and $\delta$ the noise probability.

	\begin{lemma}
		Let $S$ be a stabiliser group and let $L$ be a logical operation, i.e. an element of $Z(S) - S$ where $Z(S)$ is the centralizer of $S$. Let $U_x$ be a unitary in the Pauli group such that equation \eqref{eq:Ux-relation} holds. Then $L U_x$ also satisfies \eqref{eq:Ux-relation}. \label{lemma:unitaryambig}
	\end{lemma}

	\begin{proof}
		$L U_x P_x U_x^{\dag} L^{\dag} = L P_0 L^{\dag} = P_0$ since $[L, P_0] = 0$ and $L$ is unitary.
	\end{proof}

	\begin{remark}
		\normalfont
		Due to this ambiguity in the logical operations $U_x$, our recovery map may correct even low weight errors on an initial state $\rho$ up to some logical error $L$. However, since we have freedom in arbitrarily choosing these correction unitaries $U_x$, we simply pick them to be ones that correct errors up to weight $(d-1)/2$ without applying an extra logical operation for an error correcting code of distance $d$. We make this choice throughout the rest of the paper. $\blacksquare$ \label{remark:Uxchoice}
	\end{remark}
	We will need the following result for later.

	\begin{theorem}
		Let $N_{\delta}^{\otimes n}$ be the transfer matrix corresponding to the iid depolarizing noise channel $\mathcal{N}_{\delta}^{\otimes n}$ in \eqref{eq:iidDepolar}. Assume $d \geq 2n\delta +1$, where $d$ is the distance of the stabiliser code for the stabiliser encoded Hamiltonian in \eqref{eq:hamiltonianFull}, and suppose the correction unitaries $U_x$ satisfy the assumption in Remark \ref{remark:Uxchoice}. Then with notation as above, $\|\text{DQE}\cdot R \cdot N_{\delta}^{\otimes n} - \text{DQE}\|$ in the spectral norm is $c^{((d+1)/2)^2}$ where $c < 1$ and depends on $d$, $n$, and $\delta$. \label{thm:externalnoisedqe}
	\end{theorem}

	\begin{proof}%\phantom{\qedhere}
		We will first work in the channel representation of these operators and then relate the relevant bounds to those for the corresponding transfer matrices.

		First note that since the DQE channel is trace non-increasing, $\|\text{DQE}\| \leq 1$. From \ref{eq:transferspecnorm}, the sub-multiplicativity of the spectral norm implies that it suffices to bound $$\|R \cdot N_{\delta}^{\otimes n} - I\| = \max_{\rho \neq 0} \frac{\|\mathcal{R}(\mathcal{N}(\rho)) - \rho\|_2}{\|\rho\|_2}.$$ For any $\rho$, we have
		\begin{align}
			\mathcal{R}(\mathcal{N}^{\otimes n}_{\delta}(\rho)) &= \sum_x U_x P_x \sum_{i=0}^n (1-\delta)^{n-i} \left(\frac{\delta}{3}\right)^i \left(\sum_{s: \text{weight s} = i} P(s) \rho P(s) \right) P_x U_x^{\dag} \\
			&= \sum_{i=0}^n \sum_{{s: \text{weight s} = i}}\sum_x (1-\delta)^{n-i} \left(\frac{\delta}{3}\right)^i U_x P_x P(s) \rho P(s) P_x U_x^{\dag}.
		\end{align}
		By our assumption in Remark \ref{remark:Uxchoice}, the $U_x P_x$ operations measure the syndrome space projectors $P_x$ and perform rotations $U_x$ to rotate the components of the state $P(s) \rho P(s)$ outside of the codespace back into the codespace without an additional logical operation if the errors are weight at most $(d-1)/2$. While any stabiliser code is degenerate against iid depolarizing noise, each error still corresponds to one syndrome (different errors may still have the same syndrome, but this does not affect the subsequent calculations). There are $3^i \binom{n}{i}$ possible Pauli errors of a given weight $i$, so combining all of these facts gives
		\begin{align}
			&\|\mathcal{R}(\mathcal{N}^{\otimes n}_{\delta}(\rho))-\rho\|_2 \leq \nonumber \\
			&\left\|\left(\sum_{i=0}^{(d-1)/2} (1-\delta)^{n-i} \left(\frac{\delta}{3}\right)^i \binom{n}{i} 3^i -1 \right)\rho +\sum_{i=(\frac{d+1}{2})}^n \sum_{{s: \text{weight s} = i}} (1-\delta)^{n-i} \left(\frac{\delta}{3}\right)^i L_{i,s} \rho L_{i,s}^{\dag} \right\|_2 \nonumber \\
			&\leq \left(1- \sum_{i=0}^{(d-1)/2} \delta^i(1-\delta)^{n-i} \binom{n}{i} + \sum_{i=(d+1)/2}^n \delta^i (1-\delta)^{n-i} \binom{n}{i} \right)\|\rho\|_2, \nonumber
		\end{align}

		where the $L_{i,s}$ are logical operators, some of which may be stabilizers, and we've used that $\|\rho\|_2 = \|L_{i,s}\rho L_{i,s}^{\dag}\|_2$ since the $L_{i,s}$ are unitary and will not change the Schatten 2-norm of the density matrix.

		In the regime where $(d-1)/2 \geq n\delta$, i.e. $d \geq 2n \delta + 1$, we can bound the sums above by giving a tail bound for a binomial distribution $X$ with $n$ ``trials" and success probability $\delta$. Hoeffding's inequality applied to the binomial distribution gives $$\text{Pr}(X - E(X) \geq t) \leq  \exp(-\frac{2}{n}t^2)$$ where $E(X) = n\delta$ is the average of the binomial distribution. Thus $$\sum_{i=k}^n \delta^n(1-\delta)^{n-i} \binom{n}{i} = \text{Pr}(X - n\delta \geq k - n\delta) \leq \exp\left(-\frac{2}{n}(k - n\delta)^2\right).$$ Letting $k = (d+1)/2$ then shows that $$\text{Pr}\left(X \geq \frac{d+1}{2}\right) = \sum_{i=(d+1)/2}^n \delta^i (1-\delta)^{n-i} \binom{n}{i} \leq \exp(-\frac{2}{n}\left(1 - \frac{2n\delta}{d+1}\right)^2)^{(\frac{d+1}{2})^2}.$$

		Note that $\text{Pr}(X \leq k-1) = 1 - \text{Pr}(X \geq k)$.  Again making the substitution $k = (d+1)/2$ gives $$-\sum_{i=0}^{(d-1)/2} \delta^i (1-\delta)^{n-i} \binom{n}{i} = -1 + \text{Pr}\left(X \geq \frac{d+1}{2}\right).$$ Setting $$c = \exp(-\frac{2}{n}\left(1 - \frac{2n\delta}{d+1}\right)^2),$$ it is immediate that $c < 1$ since $e^{-x} < 1$ for $x > 0$. We then have $$\|\mathcal{R}(\mathcal{N}^{\otimes n}_{\delta}(\rho))-\rho\|_2 \leq c^{((d+1)/2)^2} \|\rho\|_2$$ which implies
		\begin{equation*}
			\|R \cdot N_{\delta}^{\otimes n} - I\| \leq c^{((d+1)/2)^2}. \qedhere
		\end{equation*}
	\end{proof}

	\begin{remark}
		{\normalfont Since the instruments act on an enlarged Hilbert space which contains the codespace as a subspace, we cannot directly apply the fault resilience theorems from \cite{cubitt2023dissipative}. To reduce our present scenario to the one considered there, we consider the following procedure. 
        
        First note that since our stabiliser-encoded Hamiltonian $H$ is assumed to be geometrically-local, the unitary part of the AGSP $K$ can be implemented as a constant depth circuit of depth $D_K$ with measurements deferred to the end. In transfer matrix form, we can therefore denote the unitary part of \eqref{eq:DQEtransfer} as $$\text{DQE} = \hat{U}_{D_K} \cdots \hat{U}_1,$$ where $\hat{U}_j$ denote the transfer matrices corresponding to the unitary gates $U_1, \ldots, U_{D_K}$ in the circuit corresponding to DQE. We then define a faulty implementation of DQE with the operator $R \cdot N_{\delta}^{\otimes n}$ applied after every gate in $\text{DQE}$, i.e. $$\text{DQE}' = R N_{\delta}^{\otimes n} \hat{U}_{D_K} \cdots R N_{\delta}^{\otimes n} \hat{U}_1.$$
        Now consider the process $$R \cdot N_{\delta}^{\otimes n} \cdot (\text{DQE}' \cdot R \cdot N^{\otimes n}_{\delta})^m,$$ where the $\cdot$ denotes matrix multiplication. Note however that since $R = \hat{P}_0 R$, where $\hat{P}_0 = P_0 \otimes \bar{P}_0$,
			\begin{align*}
				&R \cdot N_{\delta}^{\otimes n} \cdot (\text{DQE}' \cdot R \cdot N^{\otimes n}_{\delta})^m \\
				&= \hat{P}_0 \cdot R \cdot N_{\delta}^{\otimes n} (\text{DQE}' \cdot R \cdot N^{\otimes n}_{\delta}) \cdot (\text{DQE}' \cdot R \cdot N^{\otimes n}_{\delta}) \cdots (\text{DQE}' \cdot \hat{P}_0 R \cdot N^{\otimes n}_{\delta})\\
				&= \hat{P}_0 (R \cdot \text{DQE}'') \cdots (R \cdot \text{DQE}'') \hat{P}_0 (R \cdot N_{\delta}^{\otimes n})
			\end{align*}
			where in the last line, we have defined $\text{DQE}'' = N_{\delta}^{\otimes n} \cdot \text{DQE}'$.

			When acting on an input state $\ket{\rho}$, we can absorb the first $R \cdot N_{\delta}^{\otimes n}$ operation into the definition of the state and reinterpret this procedure as acting on some initial state $\ket{\rho'} = R \cdot N_{\delta}^{\otimes n} \ket{\rho}$ which lies in the codespace due to the property that $R = \hat{P}_0 R$. Then this procedure corresponds to iterating $\hat{P}_0(R \cdot \text{DQE}'') \hat{P}_0$ within the codespace. $\blacksquare$ }\label{remark:newalg}
	\end{remark}

	\begin{proposition}
        \label{prop:dqeError}
		Define the unitary part of the \text{DQE} operation in \eqref{eq:DQEtransfer} in transfer matrix form by $$\text{DQE} = \hat{U}_{D_K} \cdots \hat{U}_1$$ as in Remark \ref{remark:newalg}, where $D_K$ is the constant depth of the circuit for the AGSP $K$ and each $\hat{U}_i = U_i \otimes \bar{U}_i$. Let \text{DQE'} be a faulty-implementation of DQE with the operator $R \cdot N_{\delta}^{\otimes n}$ applied after every unitary, i.e. $$\text{DQE'} = R N_{\delta}^{\otimes n} \hat{U}_{D_K} \cdots R N_{\delta}^{\otimes n} \hat{U}_1$$ and let $$\text{DQE}'' = N_{\delta}^{\otimes n} \cdot \text{DQE}'.$$ Then with the same assumptions as in Theorem \ref{thm:externalnoisedqe}, $$\|\hat{P}_0 (R \cdot \text{DQE}'') \hat{P}_0 - \hat{P}_0(\text{DQE})\hat{P}_0\| \leq O(D_K c^{(\frac{d+1}{2})^2}),$$ where $\hat{P}_0 = P_0 \otimes \bar{P}_0$.
	\end{proposition}
	\begin{proof}
		Note that $$\|\hat{P}_0 (R \cdot \text{DQE}'') \hat{P}_0 - \hat{P}_0(\text{DQE})\hat{P}_0\| \leq \|R \cdot \text{DQE}'' - \text{DQE}\| = \|R \cdot N_{\delta}^{\otimes n} \cdot \text{DQE}' - \text{DQE}\|.$$ Since $\|\hat{U}_j\| = 1$, we have by Theorem \ref{thm:externalnoisedqe} that
		\begin{align*}
			\|R \cdot N_{\delta}^{\otimes n} \cdot \text{DQE}' - \text{DQE}\| &= \|R \cdot N_{\delta}^{\otimes n} \cdot \text{DQE}' + R \cdot N_{\delta}^{\otimes n} \cdot \text{DQE} - R \cdot N_{\delta}^{\otimes n} \cdot \text{DQE} - \text{DQE}\| \\
			&\leq \|\text{DQE}' - \text{DQE}\| + \|R \cdot N_{\delta}^{\otimes n} \cdot \text{DQE} - \text{DQE}\| \\
			&\leq \|R N_{\delta}^{\otimes n} \hat{U}_{D_K} \cdots R N_{\delta}^{\otimes n} \hat{U}_1 - \hat{U}_{D_K} \cdots \hat{U}_1 \| + c^{((d+1)/2)^2} \\
			&\leq \sum_{i=1}^{D_K} \|R N_{\delta}^{\otimes n} \hat{U}_i - \hat{U}_i\| + c^{((d+1)/2)^2} \\
			&\leq \sum_{i=1}^{D_K} \|R \cdot N_{\delta}^{\otimes n}  - I\| + c^{((d+1)/2)^2} \\
			&= O(D_Kc^{((d+1)/2)^2}). \qedhere
		\end{align*}
	\end{proof}

	\begin{remark}
		{\normalfont
			We will henceforth make the reasonable assumption that the syndrome decoding and determination of which correction unitaries $U_x$ to apply in the recovery channel are classical and a free resource. We cannot assume however that the measurements of $P_x$, i.e. the syndrome extraction portion of $R$, and the application of the correction unitaries $U_x$ themselves can be performed perfectly.

			Since we assume our Hamiltonian is geometrically local however, its ground space is embedded in a geometrically local stabiliser code and such codes are in fact quantum LDPC codes by definition. Quantum LDPC codes have constant depth syndrome extraction circuits, so the syndrome extraction part of $R$ can be implemented as a constant depth Clifford circuit \cite{delfosse2022}. The corrections $U_x$ are merely Pauli operations and can be implemented in depth one, so we may assume the entire quantum circuit for $R$ is constant-depth.

			Syndrome extraction circuits for quantum LDPC codes can in fact be made fault tolerant, so this implies the syndrome extraction in $R$ can be implemented as a constant depth fault-tolerant circuit with measurements deferred to the end \cite{gottesmanLDPC}. If we do not give a fault-tolerant implementation of the syndrome extraction in $R$ in our algorithm, then we are in the situation described by the following lemma. \textit{Whenever we refer to $R$ in what follows, we will mean the quantum circuit part of $R$ and not the classical decoding portion}. $\blacksquare$}
		\label{remark:recovery}
	\end{remark}

	\begin{lemma}
		Let $S$ be a geometrically local stabiliser code and $R'$ a noisy version of $R$ where $N_{\delta}^{\otimes n}$ is applied after every gate in $R$ (see Remark \ref{remark:recovery}). Let $P$ be the transfer matrix of the Pauli noise channel $\mathcal{P}$ that arises from commuting all $N_{\delta}^{\otimes n}$ past the gates in $R$. Then $R' = R \cdot P$ and the probability of an error of weight $w$ occurring in $P$ is at most $(1-p)^{w/2D_R} (p/3)^{w/2D_R}$, where $D_R$ is the (constant) quantum circuit depth of the syndrome extraction part of $R$. \label{lemma:recovery}
	\end{lemma}

	\begin{proof}
		From Remark \ref{remark:recovery}, the gates in  the syndrome extraction part of $R$ are Clifford and the Clifford group is by definition the normalizer of the Pauli group. Pauli errors are therefore mapped to Pauli errors with possibly larger weight when commuting $N_{\delta}^{\otimes n}$ past the Clifford gates in $R$. Thus $R'$ can equivalently be viewed as $R' = R \cdot P$ with $P$ as defined in the lemma statement. Note that since the corrections $U_x$ are merely Pauli operations, the Pauli errors will not grow in size when commuted through this part of $R$. It therefore suffices to consider the sydrome extraction part of $R$ in what follows.

		Let $w$ be the weight of some Pauli error $E$ that occurs with probability $(1-p)^{n-w} (p/3)^w$ in $\mathcal{N}_{\delta}^{\otimes n}$. Since $S$ is assumed to be geometrically local, the sydrome extraction part of $R$ has some constant circuit depth $D_R$ and we can assume its gates are nearest neighbor. To determine the weight the error $E$ can grow to after propagating through the circuit for $R$, consider the case of a single qubit error. After commuting it through the first layer of gates in the circuit, this error can grow to one of size at most 2 after passing through a CNOT or CZ gate. Commuting the size 2 error through the second layer gives a size 4 error at most if each error passes through separate CNOT or CZ gates. Suppose by induction that the size of the error at layer $i$ of the circuit is at most $2i$. This clearly holds for the first two cases considered above. At the $i$-th layer, the outermost 2 errors propagate to errors each of size at most 2 giving an error of size at most $2i-2+4 = 2i+2 = 2(i+1)$. This establishes the claim and shows that a single qubit error grows to size at most $2D_R$ after passing through the syndrome extraction part of $R$.

		Then for a weight $w$ error $E$, the worst case scenario is when each of the single qubit errors in $E$ are separated and grow to non-overlapping errors of size $2D_R$. Thus the overall weight of this commuted through error is at most $2wD_R$. This implies that the worst-case probability of a weight $w$ Pauli error occurring in the new distribution $P$ is $(1-p)^{n-w/2D_R}(p/3)^{w/2D_R}$.
	\end{proof}

	We can now apply the fault-resilience theorems of \cite{cubitt2023dissipative} to our current scenario.

	\begin{theorem}
		\label{thm:noisysDQE}
		Let $K$ be a $(\Delta,\Gamma,\varepsilon)$-AGSP for the ground space projector $\Pi_0$ of the geometrically local, stabiliser-encoded Hamiltonian in \eqref{eq:hamiltonianFull} on a Hilbert space of dimension $2^n$ and $\{\mathcal{E}_0,\mathcal{E}_1\}$ the quantum instruments of \eqref{eq:E0instrument}-\eqref{eq:E1instrument}. Assume the stabiliser code has distance $d$. Using the same notation and assumptions as in Remarks \ref{remark:newalg}-\ref{remark:recovery} and Lemma \ref{lemma:recovery}, let $\{\mathcal{E}'_0,\mathcal{E}'_1\}$ be quantum instruments such that $E'_0 = \hat{P}_0(R' \cdot \text{DQE}'')\hat{P}_0$ and $\mathcal{E}'_1(\rho) \geq \mu I$ for any state $\rho$. Consider the stopped process where iterate $\{\mathcal{E}'_0,\mathcal{E}'_1\}$ starting from the maximally mixed state on the codespace $\rho_0 = P_0/2^{n-k}$ and stop when we have obtained a run of $m$ 0's.
		\begin{enumerate}[(a)]
			\item Let $D_K$ be the constant circuit depth for the AGSP $K$, $D_R$ the constant quantum circuit depth for the syndrome extraction in $R$, and $$\alpha = \exp\left(-\frac{2}{n}\left(1 - \frac{4n\delta D_R}{d+1} \right)^2\right).$$ If $d \geq 4n\delta D_R - 1$ and $$O\left(D_K \alpha^{((d+1)/4D_R)^2}\right) < \frac{1}{2} \sqrt{\Gamma}(\sqrt{\Gamma} - \sqrt{\Delta}),$$ then the state $\rho'_m$ at the stopping point is such that 
            \begin{equation}
                \label{eq:mainresultFidelity}
                \lim_{m \rightarrow \infty} \tr(\Pi_0 \rho'_m) \geq 1 - \varepsilon - \frac{O\left(D_K \alpha^{((d+1)/4D_R)^2}\right)}{\sqrt{\Gamma}(\sqrt{\Gamma}-\sqrt{\Delta})-O\left(D_K \alpha^{((d+1)/4D_R)^2})\right)}.
            \end{equation}
            \item Suppose the recovery operation $R$ is implemented fault tolerantly. Let $$c = \exp(-\frac{2}{n}\left(1 - \frac{2n\delta}{d+1}\right)^2).$$ If $d \geq 2n\delta + 1$ and $$O\left(D_K c^{((d+1)/2)^2}\right) < \frac{1}{2} \sqrt{\Gamma}(\sqrt{\Gamma} - \sqrt{\Delta}),$$ then the state $\rho'_m$ at the stopping point is such that $$\lim_{m \rightarrow \infty} \tr(\Pi_0 \rho'_m) \geq 1 - \varepsilon - \frac{O\left(D_K c^{((d+1)/2)^2} \right)}{\sqrt{\Gamma}(\sqrt{\Gamma}-\sqrt{\Delta})-O\left(D_K c^{((d+1)/2)^2}\right)}.$$
		\end{enumerate}
	\end{theorem}

	\begin{proof}
		\begin{enumerate}[(a)]
			\item An analogous proof to that of Proposition \ref{prop:dqeError} goes through with the replacement $R \rightarrow R' = R \cdot P$ in principle, but we will then need to bound the quantity $\|R \cdot P - I \|$. From the proof of Theorem \ref{thm:externalnoisedqe}, this amounts to giving a tail bound on the probability that an error of weight at least $k$ occurs in $P$, i.e. computing $\text{Pr}(P \geq k)$.

			From the proof of Lemma \ref{lemma:recovery}, note that not every error of size $k/2D_R$ occurring in the original iid depolarizing noise model $N^{\otimes n}_{\delta}$ grows to size $k$ when propagated through the syndrome extraction part of $R$ since errors can cancel out. Since the errors in $N_{\delta}^{\otimes n}$ are binomially distributed, we therefore have $\{\text{binom} \geq k/2D_R\} \supseteq \{P \geq k\}$. In other words, the number of errors of size at least $k$ in the distribution $P$ is at most the number of errors of size at least $k/2D_R$ in the original iid depolarizing noise channel.

			For sets $A$ and $B$, $A \subseteq B$ implies $\text{Pr}(A) \leq \text{Pr}(B)$ by monotonicity, so it suffices to give a tail bound for $\text{Pr}(\text{Binom} \geq k/(2D_R))$. We then have from Hoeffding's inequality and the assumption $k/(2D_R) \geq n\delta$ that $$\text{Pr}(P \geq k) \leq \text{Pr}\left(\text{Binom} \geq \frac{k}{2D_R}\right) \leq \exp\left(-\frac{2}{n}\left(\frac{k}{2D_R} - n\delta \right)^2\right).$$ Substituting $k = (d+1)/2$ then gives $$\text{Pr}\left(P \geq \frac{d+1}{2}\right) \leq \text{Pr}\left(\text{Binom} \geq \frac{d+1}{4D_R}\right) \leq \exp\left(-\frac{2}{n}\left(1 - \frac{4n\delta D_R}{d+1} \right)^2\right)^{\left(\frac{d+1}{4D_R}\right)^2}.$$

			Letting $$\alpha = \exp\left(-\frac{2}{n}\left(1 - \frac{4n\delta D_R}{d+1} \right)^2\right),$$ the claim then follows from Theorem 61 of \cite{cubitt2023dissipative}.
            \item From Remark \ref{remark:recovery}, the syndrome extraction part of $R$ can be implemented fault tolerantly since the code is quantum LDPC. The claim then follows from Proposition \ref{prop:dqeError} and Theorem 61 of \cite{cubitt2023dissipative}. \qedhere
		\end{enumerate}
	\end{proof}

	These results should be compared to the standard DQE result of $$\lim_{m \rightarrow \infty} \tr(\Pi_0 \rho'_m) \geq 1-\varepsilon-\frac{2\delta}{\sqrt{\Gamma}(\sqrt{\Gamma} - \sqrt{\Delta})-2\delta}$$ under a noise rate $\delta$ (see Theorem 61 of \cite{cubitt2023dissipative}). Recall that $\varepsilon$ here is related both to the parameters of the AGSP and the strength of the weak measurements used (see Theorem \ref{thm:AGSP}) and can be made arbitrarily small. Thus in the presence of noise of rate $\delta$, the limiting factor in achieving the best possible fidelity is the noise rate itself.
    
    The significance of the above theorem is therefore twofold: 
    \begin{enumerate}[(a)]
        \item We can suppress a part of the additive error in final state output by the algorithm exponentially in the code distance, which is not something a naive application of the normal DQE algorithm can achieve. Thus we can get closer to fault-tolerance in the task of preparing ground states for stabiliser-encoded Hamiltonians \textit{without} the overhead required to make the entire algorithm fault-tolerant. Even if we opt to make the syndrome extraction part of the stabiliser recovery map $R$ fault-tolerant, this is still a much smaller overhead to pay compared to making the entire DQE procedure fault-tolerant since our stabiliser codes are qLDPC by assumption (see Remark \ref{remark:recovery}).
        \item Solving the inequalities $$d \geq 4n\delta D_R - 1$$ and $$O\left(D_K \alpha^{((d+1)/4D_R)^2}\right) < \frac{1}{2} \sqrt{\Gamma}(\sqrt{\Gamma} - \sqrt{\Delta}),$$ for $\delta$ gives the inequalities
        $$\delta \leq \frac{d+1}{4n D_R}$$ and $$\delta < O\left(\frac{d+1}{4nD_R}\left(1 + \sqrt{-\frac{n}{2} \ln \left(\frac{\sqrt{\Gamma}(\sqrt{\Gamma}-\sqrt{\Delta}}{2D_K} \right)}\right) \right).$$
        Taking the tighter upper bound for $\delta$ given by the first inequality therefore implies that if the noise rate afflicting our system is below this $\delta$, the algorithm inherently acquires error-resilience without any additional fault-tolerance overhead required. 
    \end{enumerate}   
    These points also show that it is necessary to incorporate the stabilizer structure of the Hamiltonian in the design of our algorithm to attain this additional noise-resilience; leveraging dissipation alone does not give such a guarantee. 

\section{Dissipative Quantum Computation} \label{section:dqc}
	A key desideratum of any model of quantum computation is robustness to errors and noise.
	Compatibility with existing error correction and fault-tolerant methods is also useful---for example, whilst adiabatic quantum computation is universal and therefore as computationally powerful as the circuit model, it is not clear how to make it fault-tolerant.
        In \cite{VWC09}, the authors introduced a new model of universal quantum computation: dissipative quantum computation (DQC). In the DQC model, quantum computation is not carried out by applying a sequence of gates to a quantum state, as in the circuit model. Instead, a dissipative, time-homogeneous dynamics is engineered, dependent on the desired quantum computation, such that the output of that quantum computation appears in the steady-state of this dissipative dynamics. (The original \cite{VWC09} paper engineers a continuous-time dynamics described by a Lindbladian, but directly analogous discrete-time versions thereof, described by an iterated completely-positive map, are also possible -- see below.)
        They prove that this dynamics not only converges to a steady-state that encodes the output of the computation, but that the convergence to this steady state is fast: it converges in a time that scales only polynomially with system-size. (This is non-trivial, as in general dissipative dynamics can take exponential-time or worse to converge to their steady-states.) This establishes that DQC is equivalent to the circuit model, in the strong quantum Church-Turing sense.

        Interestingly, and one of the reasons this result attracted significant attention, is that DQC does not require any state initialisation; because the output is encoded in the steady-state of a dissipative dynamics it will produce the correct output when started from an arbitrary input state.
        This also provides a measure of inherent robustness to error: even if the quantum state is completely scrambled at some point during the computation, the DQC dynamics will still re-converge it to the correct output.
        Therefore, whilst DQC could be made fault-tolerant with standard methods by implementing the DQC dynamics using a fault-tolerant quantum circuit, the hope is that it is inherently more robust to errors than the circuit model to some degree without the full overhead of quantum error-correction and fault-tolerance. 

        To analyse the noise-resilience of the DQC model, we will first show that the continuous-time dynamics described in \cite{VWC09} can be reduced by a sequence of steps to a simpler, discrete-time dynamics. This discrete-time dynamics is strictly stronger than the original continuous-time dynamics, in the following sense: it converges to the identical fixed-point; it converges to this fixed-point at least as quickly; and it is at least as robust to perturbation as the original dynamics.
        But this discrete-time dynamics will turn out to have a simple operational interpretation: it describes performing a classical random walk on the corresponding quantum circuit, where we probabilistically step forward or backward through the circuit until reaching the end.
        Thus dissipative quantum computation, whilst appearing quite different to the circuit model, is in fact just the standard circuit model in disguise.

	This leads to our main observation: the error tolerance of DQC is no better than the error tolerance of the standard quantum circuit model, irrespective of the chosen computation or its width or depth.
	Essentially, this stems from the fact that in the equivalent circuit random walk picture of DQC that we derive, errors can be corrected only at the beginning of the circuit, whilst the computational outcome can be read out only at the end.
        Therefore, since a random classical walk over the circuit takes at least as many steps to traverse the circuit as simply applying the gates in sequence in the usual circuit model (i.e.\ $T$), the computation will accumulate at least as many errors.
        (In fact, because the random walk takes $O(T^2)$ steps on average to reach the end of the circuit, it will generally accumulate \emph{more}.)
    
	In more detail, the dissipative model of quantum computation of \cite{VWC09} is constructed in the following way.
        For any given quantum computation $U = U_{T-1} \dots U_1 U_0$ on $N$ qubits, where each unitary acts non-trivially on at most two qubits, the continuous time dissipative quantum computation algorithm consists of evolving an arbitrary state $\rho$ under a Lindbladian $\mathcal{L}$ for poly($T$) time.
        The Lindbladian $\mathcal{L}$ is constructed as:
	\begin{equation}
		\mathcal{L}(\rho) := \sum_k L_k \rho L^\dagger_k -\frac{1}{2}\{L^\dagger_k L_k, \rho\}
	\end{equation}
	where
	\begin{align}
		L^{\text{init}}_i &= \ket{0}_i\bra{1} \otimes \ket{0}_\text{cl}\bra{0} \\
		L^{\text{prop}}_t &= U_t \otimes \ket{t+1}_\text{cl}\bra{t} + U^\dagger_t \otimes \ket{t}_\text{cl}\bra{t+1},
	\end{align}
	indexing over $i\in\{1,\dots,N\}$ and $t\in\{0,\dots,T-1\}$.
    $\mathcal{L}$ treats the full Hilbert space as consisting of a computation register on $N$ qubits and a clock register on one qudit of dimension $T+1$.
	By taking the further step of encoding the clock register in unary using the standard Kitaev construction, $\mathcal{L}$ can be made 5-local and defined solely on qubits (and it can be simplified further using well-known Hamiltonian complexity constructions~\cite{kempe2005complexitylocalhamiltonianproblem,oliveira2008complexityquantumspinsystems}.)

	\cite{VWC09,kastoryano_phd} proved that any state evolved under $\mathcal{L}$ converges to the unique stationary state $\rho_\infty$ in time $O(NT^3 \log T)$, where $\rho_\infty:=\frac{1}{T+1} \sum_{t=0}^T \ket{\psi_t}\bra{\psi_t} \otimes \ket{t}\bra{t}$ for $\ket{\psi_t}:=U_t \dots U_1 \ket{0}$.
	Since $\rho_\infty$ has $\frac{1}{T+1}$ support on the output $\ket{\psi_T}\bra{\psi_T}$ of the quantum computation described by the sequence of unitaries $U = \prod_t U_t$, this output can be extracted with probability $\frac{1}{T+1}$ by measuring the clock register.

    We will show that the intuition that DQC is more resilient to noise than the unitary circuit model because of guaranteed convergence to $\rho_\infty$ regardless of a mid-algorithm error fails when there is a positive rate of error above a certain threshold. Furthermore, this threshold is no better than the rate of error that can be tolerated in the raw circuit model (without any error correction or fault tolerance). Thus, unlike in the case of ground state preparation (Sec. \ref{section:dqestab}), for general-purpose quantum computation dissipative quantum computation is no more inherently robust to noise than the standard circuit model; implementing dissipative quantum computation in practice would require the same error-correction and fault-tolerance methods as the circuit model, with the same concomitant large overheads in physical qubit number and gate count.

\subsection{From continuous-time to discrete-time DQC}
        We first reduce the continuous-time dissipative quantum computation dynamics to an equivalent---but simpler to analyse---discrete-time dynamics, in a sequence of steps.
	We will do this via another continuous-time dynamics, $\mathcal{L}^{\text{global}}$, defined as having the following jump operators:
	\begin{align}
		L^{\text{global,init}}_a &= \frac{1}{\sqrt{2}}\ket{0}\bra{a}\otimes \ket{0}\bra{0} \\
		L^{\text{global,prop}}_t &= \frac{1}{\sqrt{2}} \left( U_t \otimes \ket{t+1}_\text{cl}\bra{t} + U^\dagger_t \otimes \ket{t}_\text{cl}\bra{t+1} \right) \\
		L^{\text{global,fin}} &= \frac{1}{\sqrt{2}} \mathds{1} \otimes \ket{T}\bra{T}
	\end{align}
	Note $\sum_\alpha L^\dagger_\alpha L_\alpha = \mathds{1} \otimes \mathds{1}$.
	Therefore in transfer matrix form,
	\begin{align}
		L^{\text{global}} &= \sum_\alpha L_\alpha \otimes \overline{L}_\alpha - \frac{1}{2} L^\dagger_\alpha L_\alpha \otimes \mathds{1} - \frac{1}{2} \mathds{1} \otimes L^\dagger_\alpha L_\alpha \\
		&= \sum_\alpha L_\alpha \otimes \overline{L}_\alpha - \mathds{1} \\
		&= D - \mathds{1}
	\end{align}
	where $D=\sum_\alpha L_\alpha \otimes \overline{L}_\alpha = \sum_\alpha D_\alpha \otimes \overline{D}_\alpha$ is the transfer matrix of the CPTP map $\mathcal{D}$ which has Kraus operators exactly identical to the jump operators of $\mathcal{L}^{\text{global}}$.

	Now we show that $\mathcal{L}^{\text{global}}$ is a strengthening of $\mathcal{L}$, in the sense that it converges at least as fast as $\mathcal{L}$ and to the same fixed-point.
	As in \cite{VWC09}, wlog all unitaries can be set to $\mathds{1}$, and after conjugation with $X' = X_0^{\otimes N} \otimes \sum_t \ket{tt}\bra{tt} + \mathds{1} \otimes (\mathds{1} - \sum_t \ket{tt}\bra{tt})$, where
	\begin{equation}
		X_0 := \ket{00}\bra{00} + \ket{00}\bra{11} - \ket{11}\bra{11} + \ket{01}\bra{01} + \ket{10}\bra{10}
	\end{equation}
	we have
	\begin{align}
		X'L^{\text{global}}{X'}^{-1}&= \frac{1}{2} \left( \ket{00}\bra{00} \otimes \ket{00}\bra{00} + \mathds{1} \otimes \ket{TT}\bra{TT} + \sum_t\mathds{1} \otimes \phi_t \right) - \mathds{1}
	\end{align}
	where
	\begin{align}
		\phi_t := \begin{pmatrix}
			0 & 0 & 0 & 1 \\
			0 & 0 & 1 & 0 \\
			0 & 1 & 0 & 0 \\
			1 & 0 & 0 & 0
		\end{pmatrix}
	\end{align}
	The matrix is now block-diagonal with: one-dimensional blocks $(0)-\mathds{1}$; two-dimensional blocks
	\begin{align}
		\begin{pmatrix}
			0 & \frac{1}{2} \\
			\frac{1}{2} & 0
		\end{pmatrix} - \mathds{1}
	\end{align}
	and $T+1$ dimensional blocks
	\begin{align} \label{lglobaltoeplitz}
		\frac{1}{2}\begin{pmatrix}
			1 & 1 & 0 & \cdots & 0 & 0 \\
			1 & 0 & 1 & \cdots & 0 & 0 \\
			0 & 1 & 0 & \cdots & 0 & 0 \\
			\vdots & \vdots & \vdots & \ddots & \vdots & \vdots \\
			0 & 0 & 0 & \cdots & 0 & 1 \\
			0 & 0 & 0 & \cdots & 1 & 1
		\end{pmatrix}
		- \mathds{1} \qquad, \qquad    \frac{1}{2} \begin{pmatrix}
			0 & 1 & 0 & \cdots & 0 & 0 \\
			1 & 0 & 1 & \cdots & 0 & 0 \\
			0 & 1 & 0 & \cdots & 0 & 0 \\
			\vdots & \vdots & \vdots & \ddots & \vdots & \vdots \\
			0 & 0 & 0 & \cdots & 0 & 1 \\
			0 & 0 & 0 & \cdots & 1 & 1
		\end{pmatrix} - \mathds{1}
	\end{align}
	Whereas for $\mathcal{L}$, we have \cite{VWC09,kastoryano_phd}: one-dimensional blocks $\left(\frac{1-\lambda_i}{4}\right) - \mathds{1}$, $\left(\frac{1-\lambda_j}{4}\right) - \mathds{1}$, $(0) - \mathds{1}$; two-dimensional blocks
	\begin{align}
		\begin{pmatrix}
			\frac{1-\lambda_i}{4} & \frac{1}{2} \\
			\frac{1}{2} & \frac{1-\lambda_j}{4}
		\end{pmatrix} - \mathds{1} \qquad , \qquad \begin{pmatrix}
			0 & \frac{1}{2} \\
			\frac{1}{2} & 0
		\end{pmatrix} - \mathds{1}
	\end{align}
	and $T+1$ dimensional blocks
	\begin{align} \label{ltoeplitz}
		\frac{1}{2}\begin{pmatrix}
			1-\frac{\lambda_i+\lambda_j}{2} & 1 & 0 & \cdots & 0 & 0 \\
			1 & 0 & 1 & \cdots & 0 & 0 \\
			0 & 1 & 0 & \cdots & 0 & 0 \\
			\vdots & \vdots & \vdots & \ddots & \vdots & \vdots \\
			0 & 0 & 0 & \cdots & 0 & 1 \\
			0 & 0 & 0 & \cdots & 1 & 1
		\end{pmatrix} - \mathds{1}
	\end{align}
	where $\lambda_\alpha \in \{0,\dots,N\}$.
    The gap of $\mathcal{L}^{\text{global}}$ can be seen from eq (\ref{lglobaltoeplitz}) to essentially arise from the difference between the largest eigenvalue of the Toeplitz matrix in the cases $\lambda_i + \lambda_j = 0$ and $\lambda_i + \lambda_j = 2$ in eq (\ref{ltoeplitz}).
    Therefore the gap of $\mathcal{L}$, which comes from the difference between the largest eigenvalue in the cases $\lambda_i + \lambda_j = 0$ and $\lambda_i + \lambda_j = 1$, is strictly smaller, and so $\mathcal{L}^{\text{global}}$ converges strictly faster than $\mathcal{L}$.

	Now, $L^{\text{global}}$ generates the semigroup
	\begin{align}
		e^{L^{\text{global}} m} = \sum_{k} \frac{m^k e^{-m}}{k!} D^k
	\end{align}
	i.e.\ it implements the discrete-time dynamics $D$ with a Poisson clock of rate~1.
	Therefore, the dynamics generated by iterating the CPTP map $\mathcal{D}$ is a discrete-time version of $\mathcal{L}^{\text{global}}$.

        Recall that $\mathcal{D}$ is defined by its Kraus operators:
        \begin{align} \label{D eqs 1}
		D^{\text{init}}_{a} &= \frac{1}{\sqrt{2}} \ket{00\dots0}\bra{a} \otimes \ket{0}_{\text{cl}}\bra{0} \\
		D^{\text{fin}} &= \frac{1}{\sqrt{2}} \mathds{1} \otimes \ket{T}_{\text{cl}}\bra{T} \\ \label{D eqs 3}
		D^{\text{prop}}_{t} &= \frac{1}{\sqrt{2}} \left( U_{t+1} \otimes \ket{t+1}_{\text{cl}}\bra{t} + U^{\dagger}_{t+1} \otimes \ket{t}_{\text{cl}}\bra{t+1} \right)
	\end{align}
	indexing over $a \in \{0,1\}^N$ and $t \in \{0,\dots,T-1\}$.

	We will show that $\mathcal{D}$ is exponentially well approximated by the simpler map $\mathcal{C}$.
	\begin{align} \label{F eqs 1}
		C^{\text{init}}_{a} &= \frac{1}{\sqrt{2}} \ket{00\dots0}\bra{a} \otimes \ket{0}_{\text{cl}}\bra{0} \\ \label{F eqs 2}
		C^{\text{fin}} &= \frac{1}{\sqrt{2}} \mathds{1} \otimes \ket{T}_{\text{cl}}\bra{T} \\ \label{F eqs 3}
		C^{\text{fw}}_{t} &= \frac{1}{\sqrt{2}} U_{t+1} \otimes \ket{t+1}_{\text{cl}}\bra{t} \\ \label{F eqs 4}
		C^{\text{bw}}_{t} &= \frac{1}{\sqrt{2}} U^{\dagger}_{t+1} \otimes \ket{t}_{\text{cl}}\bra{t+1}
	\end{align}
	indexing over $a \in \{0,1\}^{N}$ and $t \in \{0,\dots,T-1\}$.

        Operationally, $\mathcal{C}$ consists of measuring the clock register and flipping a coin.
        Then, conditional on the outcome, $\mathcal{C}$ either applies a unitary on the quantum register and increments the value on the clock register, or it applies the inverse of the previous unitary and decrements the clock. Except when the value on the clock is: a) $t=0$, where the action on one of the two outcomes is instead to reset the quantum register to $\ket{00\dots0}\bra{00\dots0}$ and leave the clock register untouched; or b) $t=T$, where the action on one of the two outcomes is to leave both the computational and clock registers untouched.
	That is, $\mathcal{C}$ essentially corresponds to a classical random walk with boundaries on the quantum circuit for $U=\prod_t U_t$.

	It has been noticed before that the dissipative part of a Lindbladian can model a classical walk as coherences vanish exponentially quickly \cite{Whitfield}.
	There, this observation was used to design a framework for interpolating between quantum walks and classical walks.
        Here, this leads to the observation that dissipative quantum computation, despite apparently being quite different to the circuit model, is in fact essentially just the standard circuit model in disguise.

	More precisely, we will show that $\mathcal{C}$ has the same eigenvectors (and, bar a few, the same eigenvalues) as $\mathcal{D}$ and converges to the unique fixed point $\rho_\infty$ in the exact same number of iterations, which is also the same time as $\mathcal{L}$ requires.

\subsection{Convergence of $\mathcal{D}$ and $\mathcal{C}$}
	By conjugating the Kraus operators with $W:=\sum_t U_t\dots U_1 \otimes \ket{t}\bra{t}$, the eigenvalues of $\mathcal{D}$ are unchanged upon setting all $U_{t+1}$ and $U^\dagger_{t+1}$ to $\mathds{1}$.
	Put the Kraus operators of $\mathcal{D}$ into transfer matrix form $D$.
	Conjugation of the transfer matrix by $X:=X_0^{\otimes N}\otimes \mathds{1}_{\text{cl}}$ preserves the spectrum (but now not the eigenvectors).
	Note $X = X^{-1}$.
	After rearranging indices,
	\begin{align}
		XDX^{-1} =
		S_N
		\oplus
		S_0 ^{\oplus 2^N-1}
		\oplus
		\begin{pmatrix}
			0 & \frac{1}{2} \\
			\frac{1}{2} & 0
		\end{pmatrix}^{\oplus 2^NT}
		\oplus 0
	\end{align}
	where the $T+1$ dimensional matrix $S_k$ is defined as
	\begin{equation}
		S_k := \frac{1}{2}\begin{pmatrix}
			\frac{k}{N} & 1 & 0 & \dots & 0 & 0 & 0 \\
			1 & 0 & 1 & \dots & 0 & 0 & 0 \\
			0 & 1 & 0 & \dots & 0 & 0 & 0\\
			\vdots & \vdots & \vdots & \ddots & \vdots & \vdots & \vdots\\
			0 & 0 & 0 & \dots & 0 & 1 & 0\\
			0 & 0 & 0 & \dots & 1 & 0 & 1 \\
			0 & 0 & 0 & \dots & 0 & 1 & 1
		\end{pmatrix}
	\end{equation}
	where $k \in \{0,1,\dots,N\}$.

	The smallest of the inter- and intra-block spectral gaps characterises the overall gap of $\mathcal{D}$.
	\begin{theorem} \label{f spec}
		The spectral gap of $\mathcal{D}$ is $\Theta\left(T^{-2}\right)$, and there is a unique eigenvector with an eigenvalue of magnitude 1.
	\end{theorem}
	\begin{proof}
		The $\begin{pmatrix}
			0 & \frac{1}{2} \\
			\frac{1}{2} & 0
		\end{pmatrix}$ blocks have eigenvalues $\pm \frac{1}{2}$.
		For $S_N$, $\lambda_1 = 1$ and $\lambda_2 = \cos \frac{\pi}{T+1} = 1 - \Theta(T^{-2})$ \cite[Thm 3.4 ii]{yueh_cheng_2008}.
		For $S_0$, $\lambda_1 = \cos \frac{\pi}{2T+3} = 1 - \Theta(T^{-2})$ \cite[Thm 3.2 viii]{yueh_cheng_2008}.
		So there is a unique eigenvector with an eigenvalue of magnitude $1$ and the next largest eigenvalue is $1-\Theta(T^{-2})$.
	\end{proof}
	The spectral gap can be used to upper bound the convergence time.
	The $\varepsilon$-convergence time $\hat{m}(\epsilon)$ for a CPTP map $\mathcal{K}$ with no non-unity eigenvalues of magnitude 1 is
	\begin{equation*}
		\hat{m}(\epsilon) := \min \left( m> 0 : \sup_\rho \norm{\mathcal{K}^m(\rho)-\mathcal{K}^\infty(\rho)}\leq \varepsilon \right)
	\end{equation*}
	If there exists $C\geq 0$, $\mu < 1$ for some CPTP map $\mathcal{K}$ with no non-unity eigenvalues of magnitude 1 such that
	\begin{equation} \label{convergence time}
		\norm{\mathcal{K}^m-\mathcal{K}^\infty}_{1\rightarrow 1}\leq C\mu^m
	\end{equation}
	is true for all integers $m>0$, then $\hat{m}=\frac{\log(\varepsilon/C)}{\log(\mu)}$ is the $\varepsilon$-convergence time as eq. (\ref{convergence time}) reduces to $\norm{\mathcal{K}^{\hat{m}}-\mathcal{K}^{\infty}}_{1\rightarrow 1}\leq \varepsilon$.
	To find $C$ and $\mu$,
	\begin{proposition}[\cite{Szehr_2013}] \label{szehr sim}
		Let $\mathcal{K}$ be a completely-positive, trace-preserving map on $M_d(\mathbb{C})$.
		If $\mathcal{K}$ has a unique stationary state and there is a similarity transformation $S$ such that $S\circ \mathcal{K}\circ S^{-1}$ is normal, then $C=\sqrt{2d}\kappa_\tau$, where $d$ is the dimension of $\mathcal{K}$, $\kappa_\tau:=\norm{S\otimes S^{-1}}_{2\rightarrow 2}$ and $\mu$ is the largest nonunity eigenvalue.
	\end{proposition}
	\begin{lemma} \label{conv bound prop}
		For $C = \sqrt{2^{N+1}(T+1)}\phi^{2N}\geq 0$ and $\mu_\mathcal{D} =1-\Theta(T^{-2}) < 1$, eq. (\ref{convergence time}) is true for $\mathcal{D}$ for all integers $m>0$, where $\phi:= (1+\sqrt{5})/2$.
	\end{lemma}
	\begin{proof}
		It is easy to check $\mathcal{D}$ is CPTP and $XWDW^{-1}X^{-1}$ is normal.
		Then $\norm{XW \otimes XW}_{2 \rightarrow 2} = \norm{X \otimes X}_{2 \rightarrow 2} \leq \norm{X_0}_{2\rightarrow 2}^{2N} \norm{\mathds{1}}_{2\rightarrow 2}^{2} = \norm{X_0}_{2\rightarrow 2}^{2N}$ by unitary invariance and submultiplicativity of the norm under the tensor product.
		By computation the largest singular value of $X_i$ is $\phi:=(1+\sqrt{5})/2$.
		The dimension of $\mathcal{D}$ is $2^N(T+1)$.
		Theorem \ref{f spec} characterises the largest nonunit eigenvalue, giving $\mu_\mathcal{D}$.
	\end{proof}
	\begin{corollary} \label{lemma O(npolyt2)}
		The $\varepsilon$-convergence time of $\mathcal{D}$ to the unique stationary state $\rho_\infty$ is
		\begin{equation*}
			\hat{m} = O(NT^2\log T\log 1/\varepsilon)
		\end{equation*}
	\end{corollary}
	\begin{proof}
		From Lemma \ref{conv bound prop},
		$C = \sqrt{2^{N+1}(T+1)}\phi^{2N}$ and $\mu_\mathcal{D} =1-\Theta(T^{-2})$.
		So
		\begin{equation*}
			\hat{m}=-\frac{\log(C/\varepsilon)}{\log(\mu)}    = \Theta(T^2 (\log C + \log 1/\varepsilon)) = O(N T^2 \log T \log 1/\varepsilon)
		\end{equation*}
		using $\frac{-1}{\log(1-\Theta\left(\frac{1}{T^{2}}\right))} = \Theta(T^2)$.
		It is easy to verify $\rho_\infty$ is a stationary state and by Theorem \ref{f spec} it is unique.
	\end{proof}

	We will now show that the classical-walk map $\mathcal{C}$ also converges to the same unique fixed point $\rho_\infty$ in the same time as $\mathcal{D}$ (and so not slower than $\mathcal{L}$).

	The above calculations for $\mathcal{D}$ go through identically for $\mathcal{C}$, except that now
	\begin{equation}
		XCX^{-1}=S_N \oplus S_0^{\oplus 2^N-1}\oplus 0.
	\end{equation}
	As a result,
	\begin{theorem}
		The spectral gap of $\mathcal{C}$ is $\Theta\left(T^{-2}\right)$, and there is a unique eigenvector with an eigenvalue of magnitude 1.
	\end{theorem}
	\begin{proof}
		See the proof of Theorem \ref{f spec}.
	\end{proof}
	\begin{corollary} \label{conv time C}
		The $\varepsilon$-convergence time of $\mathcal{C}$ to the unique stationary state $\rho_\infty$ is
		\begin{equation*}
			\hat{m} = O(NT^2\log T\log 1/\varepsilon)
		\end{equation*}
	\end{corollary}
	\begin{proof}
		See the proof of Corollary \ref{lemma O(npolyt2)}.
	\end{proof}
	Using that the eigenvectors of $\mathcal{D}$ and $\mathcal{C}$ are the same, and the eigenvalues differ only by $\pm \frac{1}{2}$ for some eigenvectors, we can show that the difference between $\mathcal{D}^m$ and $\mathcal{C}^m$ shrinks exponentially with $m$.

	\begin{lemma} \label{closeness}
		$\norm{D^m-C^m}_2 \leq \frac{2^{N+1}T}{2^{m}}$.
	\end{lemma}
	\begin{proof}
		Both $C$ and $D$ are non-defective and have the same eigenvectors and so we can write $D = \sum_k \lambda_k \ket{R_k}\bra{L_k}$ and $C = \sum_k \mu_k \ket{R_k}\bra{L_k}$ where $\bra{L_k}\ket{R_l}=\delta_{kl}$.
		\begin{align}
			\norm{D^m-C^m}_2 &= \norm{\left(\sum_k \lambda_k \ket{R_k}\bra{L_k} \right)^m-\left(\sum_k \mu_k \ket{R_k}\bra{L_k} \right)^m}_2 \\
			&= \norm{\sum_k \left(\lambda_k^m-\mu_k^m \right) \ket{R_k}\bra{L_k}}_2 \\
			&\leq \frac{2^{N+1}T}{2^{m}}
		\end{align}
		by counting the number of eigenvectors for which $\lambda_k = \pm \frac{1}{2}$ and $\mu_k = 0$.
	\end{proof}

\subsection{Noise tolerance of $\mathcal{D}$}
	A simple argument about the error resilience of $\mathcal{C}$ in conjunction with the closeness between $\mathcal{C}$ and $\mathcal{D}$ allows for a lower bound on the error resilience of $\mathcal{D}$.
	We restrict to a noise model that acts only on the computational register.
	Recall the interpretation of $\mathcal{C}$ as being a classical random walk on the circuit $U$.
	If we take the noise to be iid depolarising noise, as the walker traverses the circuit the error incurred is non-decreasing in expectation.
	For this noise model, the only decrease in error is via a reset to $\ket{00\dots0}\bra{00\dots0}$ at the leftmost end of the circuit.
	For a walker at this leftmost end, moving directly to position $T$ where the computation outcome can be read out accumulates the smallest overall error, which is exactly the action of the standard circuit model.
	Therefore the error incurred by the classical walk variant of DQC when iterated to convergence is lower bounded by the error incurred by the standard circuit model.

	Proposition \ref{proposition lb} formalises the above.
	\begin{proposition} \label{proposition lb}
		Fix a computation $U=U_{T-1} \dots U_1 U_0$.
		Define a channel $\mathcal{U}$ acting analogously to the circuit model for $U$:
		\begin{itemize}
			\item $\mathcal{U}:=\mathcal{U}_{T-1} \circ \dots \circ \mathcal{U}_1 \circ \mathcal{U}_0 $, where
			\begin{align}
				\mathcal{U}_t(\bullet)= \left( U_t \otimes \ket{t+1}\bra{t} \right) \bullet \left( U^{\dagger}_t\otimes \ket{t}\bra{t+1} \right)
			\end{align}
		\end{itemize}
		Let $\tilde{\mathcal{U}}$ be $\mathcal{U}$ subjected to iid depolarising noise on the computational register:
		\begin{itemize}
			\item $\tilde{\mathcal{U}}:=(\mathcal{N}\otimes \mathcal{I})\circ \mathcal{U}_{T-1} \circ \dots \circ (\mathcal{N}\otimes \mathcal{I})\circ \mathcal{U}_1 \circ (\mathcal{N}\otimes \mathcal{I})\circ \mathcal{U}_0 $, where $\mathcal{N}$ is an iid depolarising channel.
		\end{itemize}
		Let $\mathcal{C}$ have Kraus operators as in eqs (\ref{F eqs 1})-(\ref{F eqs 4}) for $U$, and let $\mathcal{\tilde{C}}:=(\mathcal{N}\otimes \mathcal{I}) \circ \mathcal{C}$ be $\mathcal{C}$ subjected to iid depolarising noise on the computational register.
		Then for a density matrix $\rho$,
		\begin{equation}
			\text{\upshape{tr}}\left(Q\cdot \frac{P \cdot \mathcal{\tilde{C}}^{\hat{m}}(\rho)}{\text{\upshape{tr}}\left(P \cdot \mathcal{\tilde{C}}^{\hat{m}}(\rho)\right)}\right) \geq \text{\upshape{tr}}\left(Q\cdot \frac{P \cdot \tilde{\mathcal{U}}\circ (\mathcal{N}\otimes \mathcal{I})(\ket{0}\bra{0}^{\otimes N}\otimes \ket{0}\bra{0})}{\text{\upshape tr}\left(P \cdot \tilde{\mathcal{U}} \circ (\mathcal{N}\otimes \mathcal{I})(\ket{0}\bra{0}^{\otimes N}\otimes \ket{0}\bra{0})\right)} \right)
		\end{equation}
		where $P := \mathds{1}^{\otimes N} \otimes \ket{T}\bra{T}$ and $Q := \ket{\overline{b}}\bra{\overline{b}} \otimes \mathds{1}^{\otimes N}$ where $b\in\{0,1\}$ is the correct output of $U$, and $\hat{m}$ is the convergence time of $\mathcal{C}$.
	\end{proposition}
	\begin{proof}
		The first application of $\mathcal{\tilde{C}}$ collapses the walker's clock register into a classical mixture of states.
		Consider a single element of the mixture.
		There are two cases to analyse: either the walker resets his state to all-zeros at some time $m$, or the walker never resets. For the first case, the walker is in the state $(\mathcal{N}\otimes \mathcal{I})(\ket{0}\bra{0}^{\otimes N}\otimes \ket{0}\bra{0})$.
		Without loss of generality say the walker does not reset again.
		Consider applying noisy versions of $\mathcal{U}_{t-1}$, $\mathcal{U}_t$, and $\mathcal{U}^\dagger_{t}$ sequentially, as in $\mathcal{C}$:
		\begin{align}
			\norm{\left((\mathcal{N}\otimes \mathcal{I}) \circ \mathcal{U}^\dagger_t \circ (\mathcal{N}\otimes \mathcal{I}) \circ \mathcal{U}_t \circ (\mathcal{N}\otimes \mathcal{I}) \circ \mathcal{U}_{t-1} \right) - \mathcal{U}_{t-1}} \\
			\geq \norm{\left((\mathcal{N}\otimes \mathcal{I}) \circ \mathcal{U}^\dagger_t \circ \mathcal{U}_t \circ (\mathcal{N}\otimes \mathcal{I}) \circ \mathcal{U}_{t-1} \right) - \mathcal{U}_{t-1}}\\
			= \norm{\left((\mathcal{N}^2\otimes \mathcal{I}) \circ \mathcal{U}_{t-1} \right) - \mathcal{U}_{t-1}}
		\end{align}
		This does not incur less noise than only applying the noisy version of $\mathcal{U}_{t-1}$ (as in the circuit model):
		\begin{align}
			\norm{\left((\mathcal{N}^2\otimes \mathcal{I}) \circ \mathcal{U}_{t-1} \right) - \mathcal{U}_{t-1}} \geq \norm{\left((\mathcal{N}\otimes \mathcal{I}) \circ \mathcal{U}_{t-1} \right) - \mathcal{U}_{t-1}}
		\end{align}
		Generalising, we can conclude that moving from position $t$ to position $t'$ requires passing through at least $|t'-t|$ depolarising channels.
		And so if the walker is to reach position $T$ after a reset at position $0$ he must pass through at least $T$ depolarising channels without any further resetting.

		In the second case where the walker never resets, from the proof of Corollary \ref{conv time C} we have $\hat{m} = \Omega(T^2)$.
		Therefore a walker in any initial position $t$ will pass through at least $T^2$ depolarising channels without resetting.
	\end{proof}
	Theorem \ref{D lower bound} uses Proposition \ref{proposition lb} to lower bound the error tolerance of the discrete-time dissipative quantum computation map $\mathcal{D}$.
	\begin{theorem}\label{D lower bound}
		Let $U$ and $\tilde{\mathcal{U}}$ be as in Proposition \ref{proposition lb}.
		Let $\mathcal{D}$ have Kraus operators as in eqs (\ref{D eqs 1})-(\ref{D eqs 3}) for $U$, and let $\mathcal{\tilde{D}}:=(\mathcal{N}\otimes \mathcal{I})\circ \mathcal{D}$ be $\mathcal{D}$ subjected to iid depolarising noise on the computational register.
		Then for a density matrix $\rho$,
		\begin{align}
			\text{\upshape{tr}}\left(Q \cdot \frac{P\cdot\mathcal{\tilde{D}}^{\hat{m}}(\rho)}{\text{\upshape{tr}}\left(P \cdot \mathcal{D}^{\hat{m}}(\rho)\right)}\right) \geq \text{\upshape{tr}}\left(Q\cdot \tilde{\mathcal{U}}\circ (\mathcal{N}\otimes \mathcal{I})(\ket{0}\bra{0}^{\otimes N}\otimes \ket{0}\bra{0}) \right) - O(2^{-T})
		\end{align}
		for $P, Q$ from Proposition \ref{proposition lb}, and where $\hat{m}$ is the convergence time of $\mathcal{D}$.
	\end{theorem}
	\begin{proof}
		\begin{align}
			\text{\upshape{tr}}\left(QP\cdot\mathcal{\tilde{D}}^{\hat{m}}(\rho)\right) &= \text{\upshape{tr}}\left(QP\cdot \left(\mathcal{\tilde{D}}^{\hat{m}} - \mathcal{\tilde{C}}^{\hat{m}} + \mathcal{\tilde{C}}^{\hat{m}}\right)(\rho)\right) \\
			&= \text{\upshape{tr}}\left(QP\cdot\mathcal{\tilde{C}}^{\hat{m}}(\rho) \right)
			-\text{\upshape{tr}}\left(QP\cdot \left(\mathcal{\tilde{C}}^{\hat{m}} - \mathcal{\tilde{D}}^{\hat{m}}\right)(\rho)\right) \\
			&\geq \text{\upshape{tr}}\left(QP\cdot\mathcal{\tilde{C}}^{\hat{m}}(\rho) \right)
			- \norm{\mathcal{\tilde{D}}^{\hat{m}} -\mathcal{\tilde{C}}^{\hat{m}}}_1
		\end{align}
		The first term can be lower bounded using Proposition \ref{proposition lb}.
		Then use $\text{\upshape{tr}}\left(P \cdot \mathcal{\tilde{C}}^{\hat{m}}(\rho)\right)=\text{\upshape{tr}}\left(P \cdot \mathcal{\tilde{D}}^{\hat{m}}(\rho)\right)=\text{\upshape{tr}}\left(P \cdot \mathcal{C}^{\hat{m}}(\rho)\right)=\text{\upshape{tr}}\left(P \cdot \mathcal{D}^{\hat{m}}(\rho)\right)=\frac{1}{T+1}$ and $P \cdot \tilde{\Phi}\circ(\mathcal{N}\otimes \mathcal{I})(\ket{0}\bra{0})=\tilde{\Phi}\circ (\mathcal{N}\otimes \mathcal{I})(\ket{0}\bra{0})$.

		For the second term, by the proof of Corollary \ref{lemma O(npolyt2)} we have $\hat{m} = \Omega(T^2)$.
		As the eigenvectors of non-defective matrices form a complete basis, write $D= \sum_k \lambda_k \ket{R_k}\bra{L_k}$, $C= \sum_k \mu_k \ket{R_k}\bra{L_k}$, and $N = \sum_k \alpha_k \ket{R_k}\bra{L_k}$.
		Then, using Lemma \ref{closeness} and $T = \text{poly}(N)$, we have
		\begin{align}
			\norm{\mathcal{\tilde{D}}^{\hat{m}} -\mathcal{\tilde{C}}^{\hat{m}}}_1 &\leq \sqrt{2^N(T+1)} \norm{\tilde{D}^{\hat{m}}-\tilde{C}^{\hat{m}}}_2 \\
			&= \sqrt{2^N(T+1)} \norm{(N\cdot D)^{\hat{m}}-(N\cdot C)^{\hat{m}}}_2 \\
			&= \sqrt{2^N(T+1)} \norm {\sum_k \alpha_k^{\hat{m}} (\lambda_k^{\hat{m}} - \mu_k^{\hat{m}})\ket{R_k}\bra{L_k}}_2 \\
			&\leq \sqrt{2^N(T+1)} \norm{D^{\hat{m}} -C^{\hat{m}}}_2 \leq \frac{2^{2N}T^2}{2^{cT^2}} = O\left(2^{-T}\right)
		\end{align}
		as $|\alpha_k| \leq 1$ for all $k$, using that $\norm{\mathcal{X}}_1 \leq \sqrt{\text{dim}\mathcal{X}} \norm{X}_2$ for a map $\mathcal{X}$ and its associated transfer matrix $X$.
	\end{proof}
	By Theorem \ref{D lower bound}, under a natural noise model---iid depolarising noise---and even when the noise is restricted to not act on DQC's auxiliary clock register, DQC is not more robust than the circuit model.
	Theorem \ref{D lower bound} can likely be extended to other incoherent noise models.
	Furthermore, the error tolerance of $\mathcal{D}$ is worse than we prove, given it is unlikely that a state would have full support on the leftmost node of the circuit exactly $T$ steps before the algorithm ends.

	However, some comments on the error model are in order.
	If we imagine a scenario in which errors are coherent, it could be that large cancellations could occur between the cumulative error from many steps forwards and the error from single step backwards.
	Then the circuit model, unable to take advantage of a tactical step backwards, may incur more error than a DQC-like algorithm.
	Theorem \ref{D lower bound} may not necessarily hold under such an error model.
	But as we have shown DQC is well-approximated by a classical random walk, this is no longer an observation specific to the merits of DQC as opposed to the circuit model, but rather to that of classically choosing a path along the circuit as opposed to traversing it directly.

	We also have not ruled out the possibility that a variant of $\mathcal{D}$ with, for example, different weightings on each Kraus operator may fare better.
	However probably there would still exist a variant of $\mathcal{C}$ that could mimic its behaviour.

	On the other hand, the upper bound on the convergence time can be used to bound the strength $\delta$ of (arbitrary) noise that $\mathcal{D}$ is able to tolerate.
	As a comparison, the bound $\delta = O\left(\frac{1}{NT}\right)$ holds for the circuit model.
	\begin{lemma}[\cite{Szehr_2013}] \label{szehr}
		For positive integer $m$ let $\rho_m:=\mathcal{K}^m(\rho_0)$ and $\widetilde{\rho}_m:=\widetilde{\mathcal{K}}^m(\widetilde{\rho}_0)$ be the evolution of two density matrices $\rho_0$, $\widetilde{\rho}_0$ with respect to two CPTP maps $\mathcal{K}, \widetilde{\mathcal{K}}: M_d(\mathbb{C}) \rightarrow M_d(\mathbb{C})$. If $\mathcal{K}$ has a unique stationary state and $\norm{\mathcal{K}^m-\mathcal{K}^\infty}_{1\rightarrow 1}\leq C\mu^m$ for $C\geq 0, \mu < 1$ and all positive integers $m$, then the distance between the evolved states can be bounded by
		\begin{equation*}
			\norm{\rho_m - \widetilde{\rho}_m}_1 \leq \begin{cases}
				\norm{\rho_0 - \widetilde{\rho}_0}_1 + m\norm{\mathcal{K}-\widetilde{\mathcal{K}}}_{1\rightarrow 1} & \text{\upshape{ for }} m\leq \hat{m} \\
				C\mu^m\norm{\rho_0 - \widetilde{\rho}_0}_1 + \left(\hat{m} + C\frac{\mu^{\hat{m}}-\mu^{m}}{1-\mu}\right)\norm{\mathcal{K}-\widetilde{\mathcal{K}}}_{1\rightarrow 1} & \text{\upshape{ for }} m > \hat{m}
			\end{cases}
		\end{equation*}
		where $\hat{m}:=\frac{\log(1/C)}{\log(\mu)}$.
	\end{lemma}
	\begin{corollary} \label{corollary evolution error}
		The distance between the evolution of an arbitrary initial state under $\mathcal{D}$ and the evolution of the same state under a perturbed map $\widetilde{\mathcal{D}}$, where $\norm{\mathcal{D}-\widetilde{\mathcal{D}}}_{1\rightarrow 1}\leq \delta$, is
		\begin{equation*}
			\norm{\rho_m - \widetilde{\rho}_m }_1 \leq
			\begin{cases}
				m \delta  & \text{\upshape for } m\leq \hat{m} \\
				\hat{m} \delta & \text{\upshape for } m > \hat{m}
			\end{cases}
		\end{equation*}
		where
		$\hat{m}=O(N T^2\log T)$ is the convergence time of $\mathcal{D}$.
	\end{corollary}
	\begin{proof}
		By the definition of convergence time
		\begin{align}
			C\frac{\mu^{\hat{m}}-\mu^{m}}{1-\mu} \leq C\frac{\mu^{\hat{m}}}{1-\mu}\leq \frac{1}{1-\mu}=\Theta(T^2).
		\end{align}
	\end{proof}

	\begin{theorem}
		$\mathcal{D}$ is robust to arbitrary noise for $\delta = O\left( \frac{1}{NT^3 \log T} \right)$.
	\end{theorem}
	\begin{proof}
		Again define projectors $P, Q$ onto time $T$ and computation outcome $\overline{b}$ respectively.
		For any $m > \hat{m}$, the probability of error (i.e. of measuring $\overline{b}$ on the computational register) given a measurement outcome of $T$ on the clock register is
		\begin{align}
			\tr \left( Q \cdot \frac{P \cdot \widetilde{\rho}_\infty}{\tr ( P \cdot \widetilde{\rho}_\infty) } \right) &= \tr \left( Q \cdot \frac{P \cdot (\widetilde{\rho}_\infty-\rho_\infty)}{\tr (P \cdot \widetilde{\rho}_\infty) } \right) + \frac{\tr \left( P \cdot \rho_\infty \right)}{\tr \left( P \cdot \widetilde{\rho}_\infty \right)} \tr \left( Q \cdot \frac{P \cdot \rho_\infty}{\tr ( P \cdot \rho_\infty) } \right) \\
			&\leq \frac{1}{\tr ( P \cdot \widetilde{\rho}_\infty)} \left( \tr \left (QP \cdot (\widetilde{\rho}_\infty - \rho_\infty ) \right) + \frac{1}{3}\tr ( P \cdot \rho_\infty) \right) \\
			&\leq \frac{1}{\tr ( P \cdot \widetilde{\rho}_\infty)} \left( \norm{\widetilde{\rho}_\infty-\rho_\infty}_1 + \frac{1}{3}\tr ( P \cdot \rho_\infty) \right) \\
			&\leq \frac{1}{\tr ( P \cdot \widetilde{\rho}_\infty)} \left(\hat{m}\delta + \frac{1}{3}\tr ( P \cdot \rho_\infty) \right)
		\end{align}
		where we have used $\norm{X}_1 := \sup_{0\leq \Pi \leq \mathds{1}} \tr (\Pi X)$, Corollary \ref{corollary evolution error}, and that by the definition of BQP we can take $\tr (Q \cdot \frac{P \cdot \rho_\infty}{\tr P \cdot \rho_\infty})~\leq~\frac{1}{3}$.
		\begin{align}
			\frac{1}{\tr ( P \cdot \widetilde{\rho}_\infty)} \left(\hat{m}\delta + \frac{1}{3}\tr ( P \cdot \rho_\infty) \right) \leq \frac{1}{\frac{1}{T+1}-\delta} \left( N T^2 \log T \delta + \frac{1}{3(T+1)} \right)
		\end{align}
		Choosing $\delta = O \left(\frac{1}{N T^3 \log T} \right)$ upper bounds the probability of error by a constant.
	\end{proof}

\section{Discussion and Outlook}
	\label{section:discussion}

	We have shown two main results in this paper: (1) our version of the DQE algorithm suppresses a part of the additive error in the ground state overlap of the final output state exponentially in the code distance and (2) the dissipative quantum computation algorithm is not more robust to iid depolarising noise than the circuit model. We now briefly discuss a few avenues for further investigation.

	The condition of geometric locality on the stabiliser encoded Hamiltonians we analyse for the DQE algorithm is physically motivated, but still raises the question of to what extent our results can be extended to other classes of Hamiltonians. The geometric locality assumption becomes relevant when considering circuit level depolarizing noise in our algorithm, since it ensures the unitary parts of the AGSP and stabiliser recovery $R$ operations admit constant-depth circuits (see Proposition \ref{prop:dqeError} and Lemma \ref{lemma:recovery}). Without this assumption, the depth of these operations will generally grow with the number of qubits $n$ and this worse scaling will propagate into the norm bounds we've given. Other types of Hamiltonians admit constant depth circuits for measurements of their terms however, and it would be interesting to determine the broadest class of Hamiltonians for which we can achieve similar error resilience bounds.

	As in normal error correction, the error correction procedure outlined in our algorithm actually introduces more errors than it fixes if the code distance is small. The fermionic encodings employed in quantum simulations are typically those that output low weight spin operators to reduce computational overhead, but these result in small distance codes \cite{DerbyCompactfermion}. Techniques nevertheless exist to boost distances while fixing the encoding rate between logical fermions and physical qubits; fermion-qubit mappings can also be chosen to optimise other quantities, such as the number of qubits required per fermion \cite{chien2022optimizing}. Our algorithm will integrate more effectively with some fermion-qubit mappings than others.

    We analyzed circuit level depolarizing noise in Section \ref{section:dqestab} due to its analytical tractability, but the results there can apply to more general noise models as the only assumptions needed are that the stabiliser recovery map $R$ corrects errors up to weight $(d-1)/2$ without introducing additional logical errors and that the noise acts locally. It is possible that similar if not tighter error resilience bounds can be obtained via instruments that simply measure the AGSP and resample to the maximally mixed state upon a failure. However, this procedure will have an exponentially worse run-time in general. We leave for future investigation the question of whether or not other stabiliser correction schemes can achieve better fault-resilience bounds.

    Recent work has also shown how dissipative dynamics can be used to prepare the ground state of specific families of Hamiltonians rapidly \cite{zhan2025rapidquantumgroundstate,lin2025dissipativepreparationmanybodyquantum}. As our work focuses primarily on the fidelity of the final output state as opposed to the algorithm run-time, one intriguing possibility for further research is whether both these techniques can be combined to yield a dissipative ground state preparation algorithm that rapidly converges to a more accurate final state for specific families of Hamiltonians.

    Lastly, we proved that the circuit-to-Lindbladian mapping of \cite{VWC09} gives a dynamics that is less robust to noise than implementing the circuit itself. 
    However, this mapping is not the only way of defining a Lindbladian, i.e. a set of operators $\{L_k\}$ and a self-adjoint operator describing the coherent evolution, such that the state reached after polytime evolution encodes a BQP computation outcome. 
    Investigating the noise resilience of alternative dissipative schemes \cite{RFA24} remains a promising direction for future research. 
    
	\section{Acknowledgements}

	We thank Daniel Zhang, Christopher Pattison, Dylan Airey, and Joel Klassen for many helpful discussions.
	J.P. is supported by the Engineering and Physical Sciences Research Council (grant number EP/S021582/1). T.S.C and A.R. are supported by the EPSRC Prosperity Partnership in Quantum Software for Simulation and Modelling (grant EP/S005021/1), and by the UK Hub in Quantum Computing and Simulation, part of the UK National Quantum Technologies Programme with funding from UKRI EPSRC (grant EP/T001062/1).

    \bibliographystyle{alpha}
	\bibliography{ref}
	\appendix
	\section{Background and Notation}

	We review here the notation and necessary background material used in the rest of the paper.

	\subsection{Notation}
	\label{appendix:notation}

	We use the notation $\| A \|_p$ to denote the Schatten $p$-norm of a bounded operator $A$.

	$\|A \|$ denotes the spectral norm of $A$, which is the largest singular value of $\|A\|$ and coincides with the Schatten-infinity norm (i.e. $p = \infty$). It also coincides with the induced matrix 2-norm which we denote as $$\|A\| = \|A\|_{2 \rightarrow 2} \coloneqq \sup_{v \neq 0} \frac{\|Av\|_2}{\|v\|_2},$$ where $\| \cdot \|_2$ is the Euclidean vector 2-norm. We analogously write the induced $p$-norms on matrices as $\|A\|_{p \rightarrow p}$. All of these matrix norms have the sub-multiplicativity property, e.g. $\|AB\| \leq \|A\| \|B\|$.

	It is known that the action of a completely positive (CP) map $\mathcal{E}$ on bounded operators can be written as $\mathcal{E}(\rho) = \sum_i E_i \rho E_i^{\dag}$, where the $E_i$ are operators obeying $E_i^{\dag} E_i \leq I$ and equality occurs when $\mathcal{E}$ is also trace-preserving (TP). This is called the Kraus operator representation and the operators $E_i$ are the Kraus operators of $\mathcal{E}$.

	The action of $\mathcal{E}$ can also be represented as matrix multiplication of an associated operator $E$ acting on vectorized density matrices denoted by $\text{vec } \rho = |\rho \rangle \rangle$ and we will occasionally use this throughout the paper. The vectorization of an operator involves the non-canonical isomorphism between a finite dimensional vector space and its dual given by $\bra{v} \mapsto \ket{v}$ and thus acts on basis elements of the vector space of operators as $\ket{l}\bra{m} \mapsto \ket{l} \ket{m}$. $E$ has matrix elements $\bra{i}\mathcal{E}(\ket{l}\bra{m})\ket{j}= \tr((\ket{i}\bra{j})^{\dag} \mathcal{E}(\ket{l}\bra{m}))$ so
	\begin{align*}
		\tr((\ket{i}\bra{j})^{\dag} \mathcal{E}(\ket{l}\bra{m})) &= \tr \left(\sum_{\alpha} (\ket{i}\bra{j})^{\dag}E_{\alpha}(\ket{l}\bra{m})E_{\alpha}^{\dag}\right)\\
		&= \sum_{\alpha} \bra{i} E_{\alpha} \ket{l} \bra{m} \bar{E}_{\alpha}^T \ket{j} \\
		&= \sum_{\alpha} \bra{ij} E_{\alpha} \otimes \bar{E}_{\alpha} \ket{lm}.
	\end{align*}
	Thus we can define $$E = \sum_{\alpha} E_{\alpha} \otimes \bar{E}_{\alpha},$$ and this is called the \textit{transfer matrix} of $\mathcal{E}$. It acts as $$E|\rho \rangle \rangle = |\mathcal{E}(\rho)\rangle \rangle.$$ We record the convenient fact that
	\begin{equation}
		\|E\| \coloneqq \max_{X \neq 0} \frac{\|\mathcal{E}(X) \|_2}{\|X\|_2} \label{eq:transferspecnorm}
	\end{equation}
	in terms of the Schatten $2$-norm of operators.

	\subsection{Stabiliser Codes}
	\label{appendix:stabcodes}

	Central to the construction of stabiliser codes is the \textit{Pauli group} denoted by $G_n$, where $n$ is the number of qubits. For $n=1$, it is the group generated by $$G_1 = \langle X,Z, iI \rangle$$ where $X,Z$ are the Pauli $\sigma_X$ and $\sigma_Z$ matrices respectively. For $n>1$, it is generated by all possible $n$-fold tensor products of the generators of $G_1$. An abelian subgroup $S$ of $G_n$ that does not contain $-I$ is called a \textit{stabiliser group} and is denoted by $S = \langle g_1, \ldots, g_k \rangle$, where $g_1,\ldots,g_k$ are its generators. $S$ acts upon a vector space $V$ of dimension $2^n$ and the common $+1$ eigenspace of the generators $g_1,\ldots,g_k$ is called the \textit{codespace} and has dimension $2^{n-k}$ as a subspace of $V$ \cite{mikeike}.

	Let $x = (x_1,\ldots,x_k)$ be a vector in $(\mathbb{Z}/2\mathbb{Z})^{\oplus k}$. These correspond to the measurement outcomes from measuring the $k$ stabilisers $g_i$, called the syndrome outcomes. For each $x$, we define the syndrome projectors
	\begin{equation}
		P_x = \frac{\prod_{i=1}^k (I + (-1)^{x_i} g_i)}{2^k}. \label{eq:syndromeprojs}
	\end{equation}
	The codespace is precisely $\im P_{(0,\ldots,0)}$. For ease of notation, we write $P_0 \coloneqq P_{(0,\ldots,0)}$ and define for each $s = 1, \ldots, k$ the projectors $$\Pi_s = \frac{I - g_s}{2}$$ onto the -$1$ eigenspace of $g_s$.

	We prove the following simple lemma:

	\begin{lemma}
		With notation as above, $\ker (\sum_{s=1}^k \Pi_s) = \mathrm{im } \ P_{(0,\ldots,0)}$. \label{lemma:codespaceproj}
	\end{lemma}

	\begin{proof}
		Observe $I = \Pi_s + (I - \Pi_s) = (I - g_s)/2 + (I + g_s)/2$ and $\Pi_s(I - \Pi_s) = (I - \Pi_s)\Pi_s = 0$ for every $s$.

		We now show $\text{im } \Pi_s = \ker (I - \Pi_s)$ for each $s$. To show this, let $v \in \text{im } \Pi_s$. Then $(I-\Pi_s) \Pi_s v = (I-\Pi_s) v = 0$. For the other containment, let $v \in \ker I - \Pi_s$. Then $(I - \Pi_s) v = 0$ so $\Pi_s v = v$.

		Now note that $\ker (\sum_{s=1}^k \Pi_s) = \cap_{s=1}^k \ker \Pi_s$. It is a set-theoretic fact that $\text{im }(f \circ g) \subseteq \text{im } f$, so since the projectors $\Pi_s, \Pi_{s'}$ commute for all $s,s'$, $\im P_{(0,\ldots,0)} = \im (I - \Pi_s)$ for all $s$. Therefore, if $v \in \im P_{(0,\ldots,0)}$, $v \in \im (I - \Pi_s)$ for all $s$. So $(I - \Pi_s)v = v$, which implies $\Pi_s v = 0$ for all $s$. Thus $\im P_{(0,\ldots,0)} \subseteq \cap_s \ker \Pi_s$. Now let $v \in \cap_s \ker \Pi_s$. From the result in the previous paragraph, $v \in \im (I - \Pi_s)$. But then $(I - \Pi_{s'}) (I - \Pi_s) v = (I - \Pi_{s'}) v = v$ since $v$ is also in $\ker \Pi_{s'} = \im (I - \Pi_{s'})$. Continuing inductively in this manner, we get that $P_{(0,\ldots,0)}v = v$.
	\end{proof}

\end{document}